\UseRawInputEncoding
\documentclass[10pt,journal,compsoc]{IEEEtran}
 \usepackage{cite}
\pdfoutput=1

\ifCLASSINFOpdf
\else
\fi

\usepackage{graphicx}
\usepackage{balance}  

\usepackage[ruled,lined,linesnumbered]{algorithm2e}
\usepackage{algorithmic}
\usepackage{indentfirst}
\usepackage{caption}
\usepackage{paralist}

\usepackage{amsmath,amssymb,epsfig,latexsym,graphics,subfigure}
\usepackage{times}

\usepackage{float}
\usepackage[all=normal, floats, bibnotes, wordspacing, charwidths, indent, lists]{savetrees}
\usepackage{syntonly}
\usepackage{pifont}
\usepackage{xspace}
\usepackage{textcomp}

\usepackage{colortbl}

 \newtheorem{example}{Example}
 \newtheorem{definition}{Definition}
\newtheorem{lemma-carelation}{Lemma}

 \newtheorem{theorem}{Theorem}

\renewcommand{\IEEEQED}{\IEEEQEDopen}

\begin{document}


\title{Efficient Reachability Ratio Computation for 2-hop Labeling Scheme}

\author{Junfeng~Zhou,
        Xian~Tang,
        Ming~Du
\IEEEcompsocitemizethanks{\IEEEcompsocthanksitem Junfeng Zhou and Ming Du  were with the School of Computer Science and Technology, Donghua University, Shanghai, China 201620.\protect\\
E-mail: \{zhoujf,duming\}@dhu.edu.cn.}

\IEEEcompsocitemizethanks{\IEEEcompsocthanksitem Xian Tang was with Shanghai University of Engineering Science, Shanghai, China 201620.\protect\\
E-mail: tangxian@sues.edu.cn.}


}

\markboth{Journal of \LaTeX\ Class Files,~Vol.~14, No.~8, August~2015}%
{Shell \MakeLowercase{\textit{et al.}}: Bare Demo of IEEEtran.cls for Computer Society Journals}



%
%
%

\newcounter{remark}[section]
\renewcommand{\theremark}{\nthesection.\arabic{remark}}
\newenvironment{remark}{\begin{em}
        \refstepcounter{remark}
        {\vspace{1ex}\noindent\bf Remark \theremark:}}{
        \end{em}\eop\vspace{1ex}} 

\newcommand{\proofsketch}{\noindent{\bf Proof Sketch: }}
\newcommand{\myproof}{\noindent{\bf Proof: }}

\newcommand{\nthesection}{\arabic{section}}

\newcommand{\eop}{\hspace*{\fill}\mbox{$\Box$}}

\newcommand{\stitle}[1]{\vspace{1ex} \noindent{\bf #1}}

\newcommand{\sstitle}[1]{\vspace{1ex} \noindent{\textit{ #1}}}

\newcommand{\ssstitle}[1]{\vspace{1ex} \noindent{\underline{ #1}}}

\newcommand{\kw}[1]{{\ensuremath {\mathsf{#1}}}\xspace}
\newcommand{\kwnospace}[1]{{\ensuremath {\mathsf{#1}}}}
\newcommand{\ltt}{\kw{LTT}}
\newcommand{\arr}{\kw{arrive}}
\newcommand{\vp}{\kw{p}}

\newcommand{\bellf}{{\sc Bellman-Ford}\xspace}
\newcommand{\bfalgo}{{\sc OR}\xspace}

\newcommand{\aalgo}{{\sc KDXZ}\xspace}

\newcommand{\dijk}{{\sc Dijkstra}\xspace}
\newcommand{\dalgo}{{\sc Two-Step-LTT}\xspace}

\newcommand{\dtalgo}{{\sc DOT}\xspace}

\newcommand{\genfunc}{{\sl timeRefinement}\xspace}
\newcommand{\pathc}{{\sl pathSelection}\xspace}
\newcommand{\fifo}{{\sl FIFO}\xspace}

\newcommand{\st}{starting time\xspace}
\newcommand{\sti}{starting-time interval\xspace}
\newcommand{\stsi}{starting-time subinterval\xspace}
\newcommand{\stsis}{starting-time subintervals\xspace}
\newcommand{\stis}{starting-time intervals\xspace}
\newcommand{\ti}{time interval\xspace}
\newcommand{\tis}{time intervals\xspace}
\newcommand{\ttime}{travel time\xspace}
\newcommand{\ttimea}{travel time\xspace}
\newcommand{\at}{arrival time\xspace}
\newcommand{\ata}{arrival-time\xspace}
\newcommand{\atf}{arrival-time function\xspace}
\newcommand{\ati}{arrival-time interval\xspace}
\newcommand{\ats}{arrival times\xspace}
\newcommand{\ed}{edge delay\xspace}
\newcommand{\eda}{edge-delay\xspace}
\newcommand{\edf}{edge-delay function\xspace}
\newcommand{\eds}{edge delays\xspace}
\newcommand{\wt}{waiting time\xspace}

\newcommand{\g}{\overline{g}}
\newcommand{\iend}{\tau}

\newcommand{\argmin}{\operatornamewithlimits{argmin}}

\newcommand{\myhead}[1]{\vspace{.05in} \noindent {\bf #1.}~~}
\newcommand{\cond}[1]{(\emph{#1})~}
\newcommand{\op}[1]{(\emph{#1})~}

\newcommand{\qc}{\ensuremath{Q^c}}

\newcommand{\rewrite}{\kw{XPathToReg}}

\newcommand{\upparen}[1]{\ensuremath{\mathrm{(}}{#1}\ensuremath{\mathrm{)}}}
\newcommand{\func}[2]{\funcname{#1}\upparen{\ensuremath{#2}}}
\newcommand{\funcname}[1]{\ensuremath{\mathit{#1}}}

\newcommand\AS{\textbf{as}\ }
\newcommand{\xsltsize}{\small}

\newcommand{\X}{{\cal X}}
\newcommand{\sem}[1]{[\![#1]\!]}
\newcommand{\NN}[2]{#1\sem{#2}}
\newcommand{\pcdata}{{\tt str}\xspace}

\newcommand{\exa}[2]{{\tt\begin{tabbing}\hspace{#1}\=\hspace{#1}\=\+\kill #2\end{tabbing}}}
\newcommand{\ra}{\rightarrow}
\newcommand{\la}{\leftarrow}
\newcommand{\rsa}{\_} 
\newcommand{\Ed}[2]{E_{{\scriptsize \mbox{#1} \rsa \mbox{#2}}}}
\newenvironment{bi}{\begin{itemize}
        \setlength{\topsep}{0.5ex}\setlength{\itemsep}{0ex}\vspace{-0.6ex}}
        {\end{itemize}\vspace{-1ex}}
\newenvironment{be}{\begin{enumerate}
        \setlength{\topsep}{0.5ex}\setlength{\itemsep}{0ex}\vspace{-0.6ex}}
        {\end{itemize}\vspace{-1ex}}
\newcommand{\ei}{\end{itemize}}
\newcommand{\ee}{\end{enumerate}}

\newcommand{\mat}[2]{{\begin{tabbing}\hspace{#1}\=\+\kill #2\end{tabbing}}}
\newcommand{\m}{\hspace{0.05in}}
\newcommand{\ls}{\hspace{0.1in}}
\newcommand{\beqn}{\begin{eqnarray*}}
\newcommand{\eeqn}{\end{eqnarray*}}

\newcounter{ccc}
\newcommand{\bcc}{\setcounter{ccc}{1}\theccc.}
\newcommand{\icc}{\addtocounter{ccc}{1}\theccc.}

\newcommand{\oneurl}[1]{\texttt{#1}}
\newcommand{\tabstrut}{\rule{0pt}{4pt}\vspace{-0.1in}}
\newcommand{\tabstruct}{\rule{0pt}{8pt}\\[-2ex]}
\newcommand{\stab}{\rule{0pt}{8pt}\\[-2.2ex]}
\newcommand{\sstab}{\rule{0pt}{8pt}\\[-2.2ex]}

\newcommand{\eat}[1]{}


\def\subfigcapskip{2pt}


\newcommand{\rdms}{{\sc rdbms}\xspace}
\newcommand{\sql}{{\sc sql}\xspace}
\newcommand{\dbms}{{\sc dbms}\xspace}

\newcommand{\cfig}{Fig.~}
\newcommand{\ctab}{Table~}
\newcommand{\csec}{Section~}
\newcommand{\cdef}{Definition~}
\newcommand{\cthm}{Theorem~}
\newcommand{\clem}{Lemma~}
\newcommand{\cequ}[1]{Equation~(#1)}
\newcommand{\SG}{\mathbf{SG}}
\newcommand{\SA}{\mathbf{SA}}
\renewcommand{\AA}{\mathbf{AA}}

\newcommand{\xml}{{\sl XML}\xspace}
\newcommand{\xlink}{{\sl XLink}\xspace}
\newcommand{\xpath}{{\sl XPath}\xspace}
\newcommand{\xpointer}{{\sl XPointer}\xspace}
\newcommand{\rdf}{{\sl RDF}\xspace}

\newcommand{\tcv}[1]{TC(#1)}
\newcommand{\rtcv}[1]{TC^{-1}(#1)}
\newcommand{\tc}{{\sl TC}}
\newcommand{\tcnode}{{\sl TC}}
\newcommand{\tcnodereverse}{{\sl TC$^{-1}$}}
\newcommand{\tr}{{\sl TR}\xspace}
\newcommand{\er}{{\sl ER}\xspace}
\newcommand{\bfs}{{\sl BFS}\xspace}
\newcommand{\DAG}{{\sl DAG}\xspace}
\newcommand{\DAGs}{{\sl DAG}s\xspace}
\newcommand{\yesgrail}{{\sl Yes-GRAIL}\xspace}
\newcommand{\code}{\kw{code}}
\newcommand{\sit}{\kw{sit}}
\newcommand{\psit}{{\cal P}_{sit}}
\newcommand{\yescode}{{\sl Yes-Label}\xspace}
\newcommand{\nocode}{{\sl No-Label}\xspace}
\newcommand{\entry}{\kw{entry}\xspace}
\newcommand{\yngindex}{{\sl YNG-Index}\xspace}
\newcommand{\rqrun}{{\sl RQ-Run}\xspace}
\newcommand{\citeseerx}{{\sl citeseerx}\xspace}
\newcommand{\gouniprot}{{\sl go-uniprot}\xspace}
\newcommand{\uniprot}{{\sl uniprot150}\xspace}

\long\def\comment#1{}

\newcommand{\PMR}{{\sl LPM}\xspace}
\newcommand{\CN}{\kw{CN}}
\newcommand{\CNs}{\kwnospace{CN}s\xspace}
\newcommand{\RN}{\kw{RN}}
\newcommand{\RNs}{\kwnospace{RN}s\xspace}
\newcommand{\RRN}{\kw{RRN}}
\newcommand{\RRNs}{\kwnospace{RRN}s\xspace}

\newcommand{\scc}{{\sl SCC}\xspace}
\newcommand{\sccs}{{\sl SCC}s\xspace}
\newcommand{\sccg}{\kwnospace{SCC}\textrm{-}\kw{Graph}}
\newcommand{\strongc}{\leftrightarrow}
\newcommand{\nstrongc}{\nleftrightarrow}
\newcommand{\emscc}{\kwnospace{EM}\textrm{-}\kw{SCC}}
\newcommand{\dfsscc}{\kwnospace{DFS}\textrm{-}\kw{SCC}}
\newcommand{\dfstree}{\kwnospace{DFS}\textrm{-}\kw{Tree}}
\newcommand{\len}{\kw{len}}
\newcommand{\dep}{\kw{depth}}
\newcommand{\tdep}{\kw{drank}}
\newcommand{\tlink}{\kw{dlink}}
\newcommand{\vedges}{up-edges\xspace}
\newcommand{\vedge}{up-edge\xspace}
\newcommand{\cvedge}{Up-Edge\xspace}

\newcommand{\drsscc}{\kwnospace{1P}\textrm{-}\kw{SCC}}
\newcommand{\drssccb}{\kwnospace{1PB}\textrm{-}\kw{SCC}}

\newcommand{\Bdrsscc}{\kwnospace{B}\textrm{-}\kwnospace{BR'}\textrm{-}\kw{SCC}}

\newcommand{\deprtree}{depth-ranked tree\xspace}
\newcommand{\cdeprtree}{Depth-Ranked Tree\xspace}
\newcommand{\drtree}{\kwnospace{BR}\textrm{-}\kw{Tree}}
\newcommand{\drplustree}{\kwnospace{BR}$^+$\textrm{-}\kw{Tree}}
\newcommand{\drscc}{\kwnospace{2P}\textrm{-}\kw{SCC}}
\newcommand{\updatedrank}{\kwnospace{update}\textrm{-}\kw{drank}}

\newcommand{\drtreeconstruct}{\kwnospace{Tree}\textrm{-}\kw{Construction}}
\newcommand{\drtreesearch}{\kwnospace{Tree}\textrm{-}\kw{Search}}
\newcommand{\depthrerank}{\kw{pushdown}}
\newcommand{\itrerank}{\kwnospace{iterative}\textrm{-}\kw{rerank}}
\newcommand{\drr}{\Downarrow}
\newcommand{\earlyrejection}{\kwnospace{early}\textrm{-}\kw{rejection}}
\newcommand{\earlyacceptance}{\kwnospace{early}\textrm{-}\kw{acceptance}}
\newcommand{\greduce}{\earlyacceptance}
\newcommand{\drea}{\kwnospace{1P}\textrm{/}\kw{ER}}
\newcommand{\myinf}{\kw{INF}}

\newcommand{\cedge}{\kwnospace{S}\textrm{-}\kw{edge}}
\newcommand{\cedges}{\kwnospace{S}\textrm{-}\kw{edges}}
\newcommand{\cgraph}{\kw{S\textrm{-}Graph}}
\newcommand{\csgraph}{\kwnospace{S}$^*$\textrm{-}\kw{Graph}}
\newcommand{\pushup}{\kw{pushup}}

\newcommand{\sibeynalgo}{\kw{EdgeByEdge}}
\newcommand{\sibeynalgoplus}{\kw{EdgeByBatch}}
\newcommand{\dfstrees}{\kwnospace{DFS}\textrm{-}\kw{Trees}}
\newcommand{\dfsforest}{\kwnospace{DFS}\textrm{-}\kw{Forest}}
\newcommand{\dfsstart}{\kwnospace{DFS}$^*$\textrm{-}\kw{Tree}}

\newcommand{\cdivide}{\kwnospace{Divide}\textrm{-}\kw{Star}}
\newcommand{\cdividetp}{\kwnospace{Divide}\textrm{-}\kw{TD}}
\newcommand{\cdividehie}{\kwnospace{Divide}\textrm{-}\kw{Cut}}

\newcommand{\merge}{\kw{Merge}}
\newcommand{\mydivide}{\kw{Divide}}

\newcommand{\restructure}{\kw{Restructure}}
\newcommand{\divideconquerdfs}{\kw{DivideConquerDFS}}
\newcommand{\divideconquer}{\kw{DivideConquer}}

\newcommand{\pcontract}{Contractible}
\newcommand{\ponsistent}{Consistent}
\newcommand{\ppreserve}{DFS-Preservable}
\newcommand{\intersect}{\cap}
\newcommand{\semidfs}{\kwnospace{SEMI}\textrm{-}\kw{DFS}}


\newcommand{\butc}{{\sl buTC}\xspace}
\newcommand{\butco}{{\sl buTC$^+$}\xspace}

\newcommand{\lowerb}[1]{{\kw{lb}}$(#1)$\xspace}
\newcommand{\upperb}[2]{{\kw{ub}}$(#1)$\xspace}

\newcommand{\sd}[2]{d({#1},{#2})}

\newcommand{\phoplabel}[1]{\mathcal{L}_{#1}}
\newcommand{\phopNum}[1]{N_{#1}}
\newcommand{\phopNk}[1]{n_{#1}}
\newcommand{\Sk}[1]{S_{#1}}

\newcommand{\parta}[1]{\mathcal{A}{(#1)}}
\newcommand{\partd}[1]{\mathcal{D}{(#1)}}
\newcommand{\partVA}[1]{\mathcal{P}_A{(#1)}}
\newcommand{\partVD}[1]{\mathcal{P}_D{(#1)}}

\newcommand{\tcset}[1]{TC{(#1)}}
\newcommand{\aset}[1]{A_{#1}}
\newcommand{\dset}[1]{D_{#1}}
\newcommand{\ahash}[1]{H_A^{#1}}
\newcommand{\dhash}[1]{H_D^{#1}}
\newcommand{\sid}[1]{id_A({#1})}
\newcommand{\sidd}[1]{id_D({#1})}

\newcommand{\nolabel}{{\sl No-Label}\xspace}
\newcommand{\yeslabel}{{\sl Yes-Label}\xspace}
\newcommand{\nolabels}{{\sl No-Labels}\xspace}
\newcommand{\yeslabels}{{\sl Yes-Labels}\xspace}

\newcommand{\nodelabel}[1]{L_{#1}}
\newcommand{\nolab}[1]{L^{\times}_{#1}}
\newcommand{\yeslab}[1]{L^{\checkmark}_{#1}}
\newcommand{\yeslabnew}[1]{L^{\checkmark}_{#1}}
\newcommand{\initialint}[1]{\yeslab{#1}=[s_{#1}, e_{#1}]}
\newcommand{\newint}[1]{\yeslabnew{#1}=[s_{#1}, e'_{#1}]}

\newcommand{\notesting}[2]{\emph{NoTest}(\nolab{#1}, \nolab{#2})}
\newcommand{\yestesting}[2]{\emph{YesTest}(\yeslab{#1}, \yeslab{#2})}

\newcommand{\topolevel}[1]{l_{#1}}

\newcommand{\noquery}{{\sl No-query}\xspace}
\newcommand{\yesquery}{{\sl Yes-query}\xspace}
\newcommand{\noqueries}{{\sl No-queries}\xspace}
\newcommand{\yesqueries}{{\sl Yes-queries}\xspace}

\newcommand{\reach}[1]{{\varphi}_{#1}}

\newcommand{\leafdown}{r_{leaf}^{\downarrow}}
\newcommand{\leafup}{r_{leaf}^{\uparrow}}
\newcommand{\topocount}{N_{t}}

\newcommand{\existrch}{{\sl eRch}\xspace}
\newcommand{\newrch}{{\sl IERch}\xspace}
\newcommand{\newrchplus}{{\sl IERch$^+$}\xspace}
\newcommand{\newrchone}{{\sl IERch-O$1$}\xspace}
\newcommand{\newrchtwo}{{\sl IERch-O$2$}\xspace}

\newcommand{\basicexp}{{\sl IExp}\xspace}
\newcommand{\basicnodeexp}{{\sl nodeExp}\xspace}

\newcommand{\intexp}{{\sl IExp$^+$}\xspace}
\newcommand{\nodeexp}{{\sl nodeExp$^+$}\xspace}

\newcommand{\dstree}{{\sl DLTree}\xspace}
\newcommand{\dstrees}{{\sl DLTrees}\xspace}
\newcommand{\dst}[1]{\mathcal{T}_{#1}}
\newcommand{\dt}{{\sl DT}\xspace}

\newcommand{\actprt}[1]{p_{#1}}
\newcommand{\actparent}{{{active parent}}\xspace}
\newcommand{\actparents}{{{active parents}}\xspace}
\newcommand{\actfn}{{{active free node}}\xspace}
\newcommand{\actfns}{{{active free nodes}}\xspace}
\newcommand{\fn}{{{free node}}\xspace}
\newcommand{\fns}{{{free nodes}}\xspace}
\newcommand{\cbn}{{{candidate block node}}\xspace}
\newcommand{\cbns}{{{candidate block nodes}}\xspace}
\newcommand{\bn}{{{block node}}\xspace}
\newcommand{\bns}{{{block nodes}}\xspace}
\newcommand{\cbnS}{\mathcal{C}}
\newcommand{\fnS}{\mathcal{F}}
\newcommand{\phop}{{{partial 2-hop labels}}\xspace}

\newcommand{\hnode}{{{hop-node}}\xspace}
\newcommand{\hnodes}{{{hop-nodes}}\xspace}

\newcommand{\rr}{{{reachability ratio}}\xspace}
\newcommand{\rrs}{{{reachability ratios}}\xspace}
\newcommand{\cs}{{{coverage size}}\xspace}

\newcommand{\bigG}{\mathcal{G}} 
\newcommand{\G}{G} 
\newcommand{\V}{V}
\newcommand{\E}{E}

\newcommand{\Ge}{G^e}
\newcommand{\Ve}{V^e}
\newcommand{\Ee}{E^e}

\newcommand{\Gtr}{G^{tr}} 
\newcommand{\Etr}{E^{tr}}
\newcommand{\Vtr}{V}
\newcommand{\Gtc}{G^{*}} 
\newcommand{\Etc}{E^{*}}

\newcommand{\GR}{G^{\varepsilon}}  
\newcommand{\VR}{V^{\varepsilon}}
\newcommand{\ER}{E^{\varepsilon}}

\newcommand{\GB}{G^b}

\newcommand{\Gminus}{G'}
\newcommand{\Vminus}{V'}
\newcommand{\Eminus}{E'}

\newcommand{\In}[2]{in_{{}}({#2})}
\newcommand{\InStar}[2]{in^*_{{}}({#2})}
\newcommand{\Out}[2]{out_{{}}({#2})}
\newcommand{\OutStar}[2]{out^*_{{}}({#2})}

\newcommand{\T}{T}
\newcommand{\tree}[1]{T_{#1}}
\newcommand{\treeup}[1]{T_{#1^{\utri}}}
\newcommand{\treedown}[1]{T_{#1^{\dtri}}}
\newcommand{\pmrtreeright}[1]{T_{{#1}}}
\newcommand{\pmrtreeleft}[1]{T_{\overline{#1}}}
\newcommand{\potree}[1]{\mathcal{T}_{{#1}}}

\newcommand{\query}[2]{{#1}?\rightsquigarrow{#2}}
\newcommand{\queryfalse}[2]{{#1}\not\rightsquigarrow{#2}}
\newcommand{\querytrue}[2]{{#1}\rightsquigarrow{#2}}
\newcommand{\kstepquery}[3]{{#1}\xrightarrow{?#3}{#2}}
\newcommand{\kquery}{$k$-hop reachability query\xspace}
\newcommand{\kqueries}{$k$-hop reachability queries\xspace}
\newcommand{\tht}{{\sl THT}\xspace}

\newcommand{\utri}{\vartriangle}
\newcommand{\dtri}{\triangledown}
\newcommand{\uroot}[1]{r^{\utri}_{#1}}
\newcommand{\droot}[1]{r^{\dtri}_{#1}}

\newcommand{\lb}[1]{label({#1})}
\newcommand{\lbIn}[2]{L^{{#2}}_{in}({#1})}
\newcommand{\lbOut}[2]{L^{{#2}}_{out}({#1})}
\newcommand{\lbInS}[1]{L^S_{in}({#1})}
\newcommand{\lbOutS}[1]{L^S_{out}({#1})}

\newcommand{\toporight}[1]{#1}
\newcommand{\topoleft}[1]{\overline{#1}}
\newcommand{\rankright}[2]{t_{{#2}}}
\newcommand{\rankleft}[2]{t_{\overline{#2}}}
\newcommand{\rank}[1]{t_{#1}}

\newcommand{\N}[1]{\widetilde{N}({#1})} 
\newcommand{\NCrt}[1]{N(#1)} 
\newcommand{\Vmax}{v_{\max}}
\newcommand{\C}[1]{C_{#1}}
\newcommand{\cost}[1]{N_{#1}}

\newcommand{\rankless}[1]{\prec_{t_{#1}}}

\newcommand{\nU}{n_{u_{\min}}}
\newcommand{\navg}{n_{avg}}
\newcommand{\nr}{n_r}
\newcommand{\nur}{n_{ur}}
\newcommand{\ndr}{n_{dr}}
\newcommand{\nul}{n_{ul}}
\newcommand{\ndl}{n_{dl}}

\newcommand{\cavg}{c_{avg}}

\newcommand{\InS}[1]{S_{in}({#1})}
\newcommand{\OutS}[1]{S_{out}({#1})}

\newcommand{\Stk}{S}

\newcommand{\setS}{S}

\newcommand{\storder}{{\sl DT}-order\xspace}
\newcommand{\storders}{{\sl DT}-orders\xspace}


\newcommand{\dagrdt}{{\sl DAG-Reduction}\xspace}
\newcommand{\cw}{{\sl CW}\xspace}
\newcommand{\gk}{{\sl GK}\xspace}

\newcommand{\ptr}{{\sl PTR}\xspace}
\newcommand{\dfs}{{\sl DFS}\xspace}
\newcommand{\pmrtr}{{\sl \butr}\xspace}

\newcommand{\markcnrrn}{{\sl markCNRRN}\xspace}
\newcommand{\butr}{{\sl buTR}\xspace}
\newcommand{\etr}{{\sl TR-B}\xspace}
\newcommand{\etrplus}{{\sl TR-O}\xspace}
\newcommand{\etrpplus}{{\sl TR-O$^+$}\xspace}
\newcommand{\compress}{\kwnospace{compress}$_{\textsf{R}}$\xspace}
\newcommand{\pmrer}{{\sl linear-ER}\xspace}
\newcommand{\sorter}{{\sl Sort-ER}\xspace}

\newcommand{\blrr}{{\sl blRR}\xspace}
\newcommand{\incrr}{{\sl incRR}\xspace}
\newcommand{\incrrplus}{{\sl incRR$^+$}\xspace}

\newcommand{\rhop}{{\sl R2Hop}\xspace}
\newcommand{\lorch}{{\sl LORch}\xspace}
\newcommand{\osrch}{{\sl OSRch}\xspace}

\newcommand{\grail}{{\sl GRAIL}\xspace}
\newcommand{\grl}{{\sl GRL}\xspace}
\newcommand{\grlgr}{{\sl GRL$_*$}\xspace}
\newcommand{\yesgrl}{{\sl Yes-GRAIL}\xspace}
\newcommand{\yg}{{\sl YG}\xspace}

\newcommand{\sgrail}{{\sl SGrail}\xspace}
\newcommand{\sgrl}{{\sl GB}\xspace}
\newcommand{\sgrlgr}{{\sl GB}$_*$\xspace}

\newcommand{\feline}{{\sl FELINE}\xspace}
\newcommand{\fl}{{\sl FL}\xspace}
\newcommand{\flgr}{{\sl FL$_*$}\xspace}

\newcommand{\ip}{{\sl IP$^+$}\xspace}
\newcommand{\ipgr}{{\sl IP$^+_*$}\xspace}
\newcommand{\sip}{{\sl IP$^+_{B}$}\xspace}
\newcommand{\sipgr}{{\sl IP$^+_{B*}$}\xspace}

\newcommand{\pll}{{\sl PLL}\xspace}
\newcommand{\pllgr}{{\sl PLL$_*$}\xspace}
\newcommand{\dl}{{\sl DL}\xspace}

\newcommand{\ferrari}{{\sl FERRARI-G}\xspace}
\newcommand{\fr}{{\sl FR}\xspace}

\newcommand{\tol}{{\sl TOL}\xspace}

\newcommand{\bfl}{{\sl BFL$^+$}\xspace}

\newcommand{\tf}{{\sl TF}\xspace}
\newcommand{\tfgr}{{\sl TF$_*$}\xspace}
\newcommand{\stf}{{\sl TF$_B$}\xspace}
\newcommand{\stfgr}{{\sl TF$_{B*}$}\xspace}

\newcommand{\avgcost}{\bigtriangleup}

\newcommand{\onl}{{\sl Label+G}\xspace}
\newcommand{\lab}{{\sl Label-Only}\xspace}

\newcommand{\lin}[1]{L_{in}({#1})}
\newcommand{\lout}[1]{L_{out}({#1})}

\newcommand{\largenode}{{\sl Large-Node}\xspace}
\newcommand{\interval}{{\sl Interval}\xspace}
\newcommand{\expinterval}{{\sl Expanded-Tree-Interval}\xspace}

\newcommand{\afterset}[2]{#1_{#2}}
\newcommand{\maxnode}[1]{v_{\max}(#1)}

\newcommand{\range}[1]{I_{#1}}
\newcommand{\newrange}[1]{I'_{#1}}

\newcommand{\notest}{{\sl No-Testing}\xspace}
\newcommand{\yestest}{{\sl Yes-Testing}\xspace}

\newcommand{\node}{node\xspace}
\newcommand{\nodes}{nodes\xspace}

\newcommand{\amaze}{\textsf{amaze}\xspace}
\newcommand{\kegg}{\textsf{kegg}\xspace}
\newcommand{\nasa}{\textsf{nasa}\xspace}
\newcommand{\xmark}{\textsf{xmark}\xspace}
\newcommand{\vchocyc}{\textsf{vchocyc}\xspace}
\newcommand{\mtbrv}{\textsf{mtbrv}\xspace}

\newcommand{\go}{\textsf{go}\xspace}

\newcommand{\anthra}{\textsf{anthra}\xspace}
\newcommand{\agrocyc}{\textsf{agrocyc}\xspace}
\newcommand{\ecoo}{\textsf{ecoo}\xspace}

\newcommand{\mail}{\textsf{email}\xspace}

\newcommand{\wiki}{\textsf{wiki}\xspace}
\newcommand{\lj}{\textsf{LJ}\xspace}
\newcommand{\web}{\textsf{web}\xspace}
\newcommand{\yago}{\textsf{yago}\xspace}
\newcommand{\dbp}{\textsf{dbpedia}\xspace}
\newcommand{\pubmed}{\textsf{pubmed}\xspace}
\newcommand{\gov}{\textsf{govwild}\xspace}
\newcommand{\human}{\textsf{human}\xspace}
\newcommand{\twitter}{\textsf{twitter}\xspace}
\newcommand{\citeseer}{\textsf{citeseer}\xspace}
\newcommand{\arxiv}{\textsf{arxiv}\xspace}

\newcommand{\uniptwo}{\textsf{uniprot22m}\xspace}
\newcommand{\citeseerten}{\textsf{10citeseerx}\xspace}
\newcommand{\patten}{\textsf{10cit-Patent}\xspace}
\newcommand{\goten}{\textsf{10go-unip}\xspace}
\newcommand{\citeseerfive}{\textsf{05citeseerx}\xspace}

\newcommand{\patfive}{\textsf{05cit-Patent}\xspace}
\newcommand{\unipten}{\textsf{uniprot100m}\xspace}
\newcommand{\gounip}{\textsf{go\_uniprot}\xspace}
\newcommand{\pat}{\textsf{patent}\xspace}
\newcommand{\citeseerxx}{\textsf{citeseerx}\xspace}
\newcommand{\unipfifteen}{\textsf{uniprot150m}\xspace}
\newcommand{\webuk}{\textsf{web-uk}\xspace}

\newcommand{\rone}{\textsf{1M2x}\xspace}
\newcommand{\rtwo}{\textsf{1M50x}\xspace}
\newcommand{\rthree}{\textsf{1M100x}\xspace}
\newcommand{\rfour}{\textsf{1M150x}\xspace}
\newcommand{\rfive}{\textsf{1M200x}\xspace}

\newcommand{\cwo}{{\sl CWO}\xspace}
\newcommand{\kreach}{{\sl kReach}\xspace}
\newcommand{\bfsi}{{\sl BFSI-B}\xspace}
\newcommand{\htt}{{\sl HT}\xspace}

\newcommand{\onexone}{\textsf{1M-1M}\xspace}
\newcommand{\onexfive}{\textsf{1M-5M}\xspace}
\newcommand{\onexten}{\textsf{1M-10M}\xspace}
\newcommand{\onexfifteen}{\textsf{1M-15M}\xspace}
\newcommand{\onextwenty}{\textsf{1M-20M}\xspace}
\newcommand{\tenxfifty}{\textsf{10M-50M}\xspace}
\newcommand{\twentyxhundrand}{\textsf{20M-100M}\xspace}
\newcommand{\thirtyxonefifty}{\textsf{30M-150M}\xspace}
\newcommand{\fourtyxtentyzero}{\textsf{40M-200M}\xspace}

\IEEEtitleabstractindextext{%
\begin{abstract}
As one of the fundamental graph operations, reachability queries processing has been extensively studied during the past decades. Many approaches followed the line of designing 2-hop labels to make acceleration. Considering that the index size cannot be bounded when using all nodes to construct 2-hop labels, researchers proposed to use a part of important nodes to construct 2-hop labels (\phop) to cover as much reachability information as possible. Then, we may achieve better query performance with limited index size and index construction time. However, \phop do not always perform well on different graphs.

In this paper, we focus on the problem of how to efficiently compute \rr, such that to help users determine whether \phop should be used to answer reachability queries for the given graph. Intuitively, \rr denotes the ratio of the number of reachable queries that can be answered by \phop over the total number of reachable queries involved in the given graph. We discuss the difficulties of \rr computation, and propose an incremental-partition algorithm for \rr computation. We show by rich experimental results that our algorithm can efficiently get the result of \rr, and show how the overall query performance is affected by different \phop.
Based on the experimental results, we give out our findings on whether \phop should be used to the given graph for reachability queries processing.
\end{abstract}

}

\maketitle

\IEEEdisplaynontitleabstractindextext

\IEEEpeerreviewmaketitle

\section{Introduction}
\label{section:intro}

Reachability queries processing is one of the fundamental graph operations and has been extensively studied in the literature~\cite{DBLP:conf/sigmod/AgrawalBJ89, DBLP:conf/icde/ChenC08, DBLP:conf/sigmod/ChengHWF13, DBLP:conf/soda/CohenHKZ02, DBLP:conf/sigmod/JinRDY12, DBLP:journals/tods/JinRXW11, DBLP:journals/pvldb/JinW13, DBLP:conf/sigmod/JinXRF09, DBLP:conf/icde/SeufertABW13, DBLP:conf/sigmod/TrisslL07, DBLP:conf/sigmod/SchaikM11, DBLP:conf/edbt/VelosoCJZ14, DBLP:conf/icde/WangHYYY06, DBLP:journals/pvldb/WeiYLJ14, DBLP:journals/pvldb/YildirimCZ10, DBLP:journals/vldb/YildirimCZ12, DBLP:series/ads/YuC10, DBLP:conf/edbt/ZhangYQZZ12, DBLP:conf/sigmod/ZhuLWX14, DBLP:conf/cikm/YanoAIY13, DBLP:journals/tkde/SuZWY17, DBLP:conf/sigmod/ZhouZYWCT17, DBLP:journals/vldb/ZhouYLWCT18, DBLP:conf/icde/SenguptaBRB19, DBLP:conf/sigmod/Sarisht19, DBLP:journals/access/DuYZTCZ19}.
Given a directed graph, a reachability query $\query{u}{v}$ asks whether~~there exists a directed path from node $u$ to $v$.
It can be used to Semantic Web (RDF), online social networks, biological
networks, ontology, transportation networks, etc, to answer whether two nodes have a certain connection.
It can also be used as a building brick in structured queries answering, such as XQuery\footnote{https://www.w3.org/TR/2017/REC-xquery-31-20170321} or SPARQL\footnote{https://www.w3.org/TR/rdf-sparql-query}.

To answer a given reachability query, researchers have proposed many efficient labeling schemes~\cite{DBLP:conf/sigmod/AgrawalBJ89, DBLP:conf/icde/ChenC08, DBLP:conf/sigmod/ChengHWF13, DBLP:conf/soda/CohenHKZ02, DBLP:conf/sigmod/JinRDY12, DBLP:journals/tods/JinRXW11, DBLP:journals/pvldb/JinW13, DBLP:conf/sigmod/JinXRF09, DBLP:conf/icde/SeufertABW13, DBLP:conf/sigmod/TrisslL07, DBLP:conf/sigmod/SchaikM11, DBLP:conf/edbt/VelosoCJZ14, DBLP:conf/icde/WangHYYY06, DBLP:journals/pvldb/WeiYLJ14, DBLP:journals/pvldb/YildirimCZ10, DBLP:journals/vldb/YildirimCZ12, DBLP:series/ads/YuC10, DBLP:conf/edbt/ZhangYQZZ12, DBLP:conf/sigmod/ZhuLWX14, DBLP:conf/cikm/YanoAIY13, DBLP:journals/tkde/SuZWY17, DBLP:conf/sigmod/ZhouZYWCT17, DBLP:journals/vldb/ZhouYLWCT18, DBLP:conf/icde/SenguptaBRB19, DBLP:conf/sigmod/Sarisht19, DBLP:journals/access/DuYZTCZ19} to make acceleration, among which 2-hop labeling scheme has been widely adopted and was shown to be better than others in many cases~\cite{DBLP:conf/soda/CohenHKZ02, DBLP:journals/pvldb/JinW13, DBLP:conf/sigmod/ChengHWF13,DBLP:conf/cikm/YanoAIY13, DBLP:conf/sigmod/ZhuLWX14, DBLP:conf/icde/SeufertABW13, DBLP:journals/access/DuYZTCZ19}.
Existing approaches that adopt 2-hop labels can be classified into two categories. The first kind of approaches \cite{DBLP:conf/soda/CohenHKZ02, DBLP:conf/sigmod/ChengHWF13,DBLP:journals/pvldb/JinW13, DBLP:conf/cikm/YanoAIY13,DBLP:conf/sigmod/ZhuLWX14} generate 2-hop labels based on all nodes, i.e., the 2-hop labels maintain the whole transitive closure (\tc). For these approaches, a reachability query $\query{u}{v}$ can be answered by comparing the 2-hop labels of $u$ and $v$ without graph traversal. However, the index size cannot be bounded w.r.t. the size of the input graph, and minimizing the size of 2-hop labels is NP-hard~\cite{DBLP:conf/soda/CohenHKZ02}.

Different with~\cite{DBLP:conf/soda/CohenHKZ02, DBLP:conf/sigmod/ChengHWF13,DBLP:journals/pvldb/JinW13, DBLP:conf/cikm/YanoAIY13,DBLP:conf/sigmod/ZhuLWX14}, the second kind of approaches~\cite{DBLP:conf/icde/SeufertABW13, DBLP:journals/access/DuYZTCZ19} do not generate 2-hop labels based on all nodes, but based on a few nodes with large degree. We call these nodes as \hnodes, and call 2-hop labels based on these \hnodes as \phop. 
Compared with the first kind of approaches, the index size of \phop can be bounded, and is usually much smaller than that of the first kind of approaches.
%
%
Even though \phop cannot answer all reachable queries in the whole \tc, it was shown in \cite{DBLP:conf/icde/SeufertABW13, DBLP:journals/access/DuYZTCZ19} that \phop can help improve the query performance significantly by answering most reachable queries for some graphs.

However, for some other graphs, the query performance may degenerate when using \phop~\cite{DBLP:conf/icde/SeufertABW13, DBLP:journals/access/DuYZTCZ19}. The reason lies in that the \rr of \phop changes violently for different graphs. Here, \rr means the ratio of the number of reachable queries that can be answered by \phop over the size of the \tc. Figure~\ref{graph:hopreach} shows the \rr of \phop on three graphs, from which we know that if we construct \phop using four \hnodes with large degree, then the \rr is greater than 90\% on \human and \webuk, meaning that the probability that a given reachable query $q$ can be answered by \phop is greater than 90\%, and is close to 0 on \pat meaning that the probability that $q$ can be answered by \phop is close to 0. In this case, using \phop brings us nothing but additional cost, which may degenerate the overall performance.

\begin{figure}
  \centering
\includegraphics{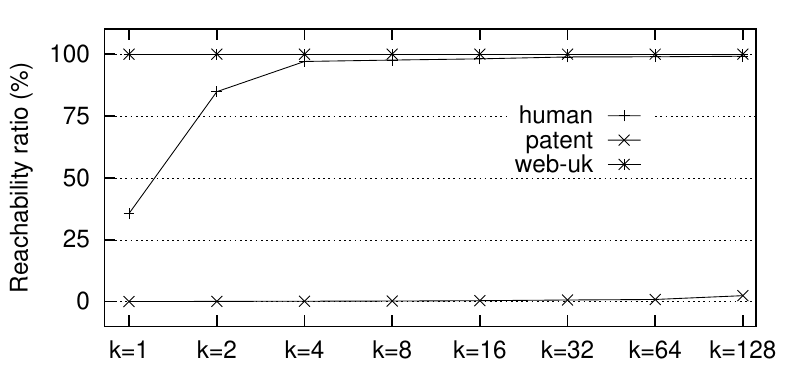}

\setlength{\abovecaptionskip}{-0.5pt}
\caption{Reachability ratio of \phop based on $k$ nodes with large degree.}
\label{graph:hopreach}
\end{figure}

Therefore, before using \phop, a key problem that needs to be solved is: how to efficiently compute the \rr of \phop w.r.t. the given graph? Because only if we know the \rr, we can determine whether we should use it. For example, given the \rr shown in Figure~\ref{graph:hopreach}, we may decide to use \phop on \human and \webuk, but not on \pat, due to that using more \hnodes on \pat cannot increase the \rr significantly. Furthermore, we can determine how many \hnodes should be chosen to construct the \phop, if we decide to use \phop to accelerate the query performance. For example, according to the \rr in Figure~\ref{graph:hopreach}, we may decide to use four \hnodes to construct \phop on \human, but for \webuk, \rr changes little with the increase of \hnodes and one \hnode is good enough, which means a larger \rr and smaller index size.

To the best of our knowledge, this is the first work that addresses the problem of \rr computation, which is not a trivial task and involves two operations. One is computing the size of \tc, the other is computing the exact number of reachable queries that can be answered by the \phop, which we call as the \cs. Considering that \tc~ size computation can be efficiently solved by existing works~\cite{DBLP:journals/cluster/TangCLL19}, the difficulty of \rr computation lies in how to efficiently compute the \cs.
The naive way is first generating \phop based on $k$ selected \hnodes, then getting the \cs by checking all node pairs using the \phop. In this way, the cost of \cs computation is $O(k|V|^2)$ and cannot scale to large graphs, where $V$ is the set of nodes in the input graph. Further, if the \rr is too small to meet the requirement, we may need to increase $k$'s value and repeat the above operation, which makes \rr computation more difficult to be solved.

We propose to compute the \cs incrementally, such that when the value of $k$ changes, we can avoid the costly \cs recomputation, such that to support efficient \rr computation.
The basic idea is, given the \cs w.r.t. $k$ nodes, when we decide to compute the \cs w.r.t. $k+1$ \hnodes, we do not compute the \cs from scratch, but only compute the increased \cs.
However, the increased \cs cannot be easily computed. To know the increased \cs w.r.t. the $(k+1)^{\text{th}}$ \hnode $u$, we need to firstly traverse from $u$ to get a set of nodes $\dset{u}$ that $u$ can reach, then traverse from $u$ backwardly to get the second set of nodes $\aset{u}$ that can reach $u$. Given $\aset{u}$ and $\dset{u}$, we need to check for each pair of nodes $(a,d)$, whether $a$ can reach $d$ can be determined by the current \phop without $u$, where $a\in \aset{u}, d\in \dset{u}$. If the answer is YES, then we know that $a$ can reach $d$ can be answered by the \phop without $u$, and should not be considered when computing the increased \cs w.r.t. $u$. The cost of processing one \hnode $u$ is as high as $O(k|\aset{u}||\dset{u}|)$. Obviously, with the increase of the number of \hnodes for \phop construction, the cost could be unaffordable.
To this problem, we propose to divide both $\aset{u}$ and $\dset{u}$ into a set of disjoint subsets based on equivalence relationship (defined later), such that for each pair of subsets $A_1\subseteq \aset{u}$ and $D_1\subseteq \dset{u}$, we only need to test one reachability query, rather than $|A_1|\times |D_1|$ queries. The cost of \rr computation is, therefore, reduced significantly even when processing large graphs.
We make the following contributions.

\begin{enumerate}
  \item To the best of our knowledge, this is the first work to address the problem of \rr computation.

  \item We propose a set of algorithms for \rr computation. We show that according to the properties of 2-hop labels, the two sets of nodes that can reach and be reached by a certain \hnode can be divided into a set of disjoint subsets, such that the computation cost can be reduced significantly. We prove the correctness and efficiency of our approach.

  \item We conduct rich experiments on real datasets. The experimental results show that compared with the baseline approach, our algorithm works much more efficiently on \rr computation. We also show how the overall query performance is affected by \phop with different number of \hnodes, based on which we give out our findings on whether \phop should be used to the given graph for reachability queries processing.
\end{enumerate}

The remainder of the paper is organized as follows. We discuss the
preliminaries and the related work in Section~\ref{section:preliminary}. In Section~\ref{section:baseline}, we give out the baseline algorithm for \rr computation, and propose the first incremental algorithm in Section~\ref{section:incremental}. After that, we propose the optimized incremental algorithm in Section~\ref{section:optimization}.
We report our experimental
studies in
Section~\ref{section:experiment}, and conclude our paper in
Section~\ref{section:conclusion}.

%
%

\section{Background and Related Work}
\label{section:preliminary}

\subsection{Preliminaries}

Given a directed graph $\bigG$,
we can construct a directed acyclic graph (\DAG) $G$ from $\bigG$ in linear
time~\cite{DBLP:journals/siamcomp/Tarjan72} by coalescing each strongly connected component (\scc) of $\bigG$ into a node in $G$. Then, the reachability query on $\bigG$ can be answered equivalently on $G$.
We follow the tradition and assume that the input graph is a \DAG.

Given a \DAG $G=(V,E)$, where $V$ is the set of nodes and $E$ the set of edges.
We define $\In{G}{u}=\{v|(v,u)\in E\}$ as the
set of in-neighbor \nodes of $u$ in $G$, and $\Out{G}{u}=\{v|(u,v)\in E\}$ the
set of out-neighbor \nodes of $u$. Similarly, we use
$\InStar{G}{u}$ to denote the set of \nodes in $G$ that can reach $u$, and $\OutStar{G}{u}$ the set of \nodes in
$G$ that $u$ can reach.
We say $u$ can reach $v$ ($\querytrue{u}{v}$), if $v\in \OutStar{G}{u}$.

The transitive closure (\tc) of $G$ is $\Gtc=(V, \Etc)$, where $\Etc=\{(u,v)|u,v\in V, v\in \OutStar{G}{u}, v\neq u\}$.
We define $\tcv{u}=\OutStar{G}{u}\setminus\{u\}$ as the transitive closure of $u$, and define $\rtcv{u}=\InStar{G}{u}\setminus\{u\}$ as the reverse \tc~ of $u$.
The \tc~ size of $G$ is denoted as $\tcv{G}=\sum_{u\in V}|\tcv{u}|$.
In~\cite{DBLP:journals/cluster/TangCLL19}, the authors proposed an efficient algorithm for \tc~ size computation with time complexity $O(r|E|)$, where $r$ is the number of distinct paths decomposed from the input
graph. In this paper, we assume that the \tc~ size is given in advance, which can be got by executing the algorithm in~\cite{DBLP:journals/cluster/TangCLL19} as an offline activity. Note that \tc~ size computation is different with \tc~ computation. The former computes $|\tcv{v}|$ for all nodes, while the latter computes $\tcv{v}$ for all nodes with time complexity $O(|V|\times |E|)$.

Given a set of $k$ nodes $\Sk{k}\subseteq V$, we use $\phoplabel{k}$ to denote 2-hop labels constructed based on nodes of $\Sk{k}$, where each node in $\Sk{k}$ is called a \hnode. If $v\in \tcv{u}$ and $\phoplabel{k}$ can correctly tell that $\querytrue{u}{v}$, we say
$\phoplabel{k}$ (or $\Sk{k}$) can cover the reachable query $\querytrue{u}{v}$.
Let $\phopNum{k}$ be the number of distinct reachable queries that can be covered by $\phoplabel{k}$, the \rr of $\Sk{k}$ is defined as Equation~\ref{eq:reachratio}.

\begin{equation}\label{eq:reachratio}
  \alpha = \phopNum{k}/\tcv{G}
\end{equation}

\stitle{Problem Statement:} Given a \DAG $\G=(V,E)$, its \tc~ size and a \hnode set $\Sk{k}\subseteq V$, return the \rr of $\Sk{k}$.

\begin{table}[t]
\caption{Notations}
\label{table:notations}

\centering


\begin{tabular} {|@{}c@{}|l|} \hline

 Notation & ~Description\\

\hline
\hline	$G=(V,E)$	          &	a \DAG with a node set $V$ and an edge set $E$		\\
\hline	$\In{G}{v}(\Out{G}{v})$	      & the set of in-neighbors (out-neighbors) of $v$  \\
\hline	$\InStar{G}{v}(\OutStar{G}{v})$	  & the set of nodes that can reach (be reached by) $v$ \\

\hline  ~$\tcv{v}(\rtcv{v})$~ & $\OutStar{G}{v}\setminus \{v\}(\InStar{G}{v}\setminus \{v\})$ \\
\hline  ~$\tcv{G}$~ & the \tc~ size of $G$ \\
\hline  $\Sk{k}$ & a set of $k$ hop nodes \\
\hline  $\lbOut{v}{k}(\lbIn{v}{k})$ & the 2-hop out (in) label of $v$ w.r.t. $\Sk{k}$ \\
\hline  $\phoplabel{k}$ & the \phop w.r.t. $\Sk{k}$ \\
\hline  $\phopNum{k}$ & the number of reachable queries covered by $\phoplabel{{k}}$ \\
\hline  $\aset{k}$ & ancestor set containing nodes that can reach $v_k$ \\
\hline	$\dset{k}$	& descendant set containing nodes that $v_k$ can reach \\

\hline
\end{tabular}
\end{table}

\subsection{Related Work}
\label{section:related}

As no existing works has addressed \rr computation, we only discuss existing works on reachability queries processing. We discuss these approaches according to whether they use 2-hop labels to answer reachability queries.

\stitle{2-hop based Approaches:}
Cohen et al. proposed to use 2-hop label~\cite{DBLP:conf/soda/CohenHKZ02} to answer reachability queries, where each node $u$ is assigned two labels, one is in-label $\lbIn{u}{}$, and the other is out-label $\lbOut{u}{}$.
$\lbIn{u}{}(\lbOut{u}{})$ consists of a set of nodes $v$ that can reach (be reached by) $u$. Given the 2-hop label, the answering of a reachability query $\query{u}{v}$ can be done by a set intersection operation on two labels, as indicated by Formula~\ref{eq:reachable}.

\begin{equation} \label{eq:reachable}
  \query{u}{v}=\left\{
    \begin{array}{ll}
    \text{TRUE}, & \lbOut{u}{}\bigcap \lbIn{v}{} \neq \emptyset,\\
    \text{FALSE},& \text{otherwise}  \\
    \end{array} \right.
\end{equation}

Existing works involving 2-hop labeling scheme can be classified into two categories. The approaches in the first category construct 2-hop labels based on all nodes~\cite{DBLP:conf/soda/CohenHKZ02,DBLP:conf/sigmod/ChengHWF13,DBLP:journals/pvldb/JinW13, DBLP:conf/cikm/YanoAIY13,DBLP:conf/sigmod/ZhuLWX14}.
Considering that minimizing 2-hop label size is NP-hard~\cite{DBLP:conf/soda/CohenHKZ02}, Cohen et al. proposed a $(\log|V|)$-approximate solution. However, the index construction cost is $O(|V||E|\log(|V|^2/|E|))$, which makes it difficult to scale to large graphs.
Motivated by this, the following works~\cite{DBLP:conf/sigmod/ChengHWF13,DBLP:journals/pvldb/JinW13, DBLP:conf/cikm/YanoAIY13,DBLP:conf/sigmod/ZhuLWX14} have to discard the approximation guarantee and focused on finding better ordering strategy to rank nodes, such that to improve the efficiency of 2-hop label construction. Even though, the index size still cannot be bounded w.r.t. the size of the input graph.

Different with the above approaches, approaches in the second category~\cite{DBLP:conf/icde/SeufertABW13, DBLP:journals/access/DuYZTCZ19} generate \phop based on a few \hnodes to cover as more reachability relationships as possible.
It was shown in \cite{DBLP:conf/icde/SeufertABW13, DBLP:journals/access/DuYZTCZ19} that \phop can work very efficiently in answering reachability queries for some graphs, due to that the \phop can cover most reachability relationships for these graphs. Moreover, the index size of \phop can be bounded, and is usually much smaller than that of the first kind of approaches in practice.
However, for some other graphs, they cannot work efficiently~\cite{DBLP:conf/icde/SeufertABW13, DBLP:journals/access/DuYZTCZ19}. The query performance may degenerate due to small \rr for these graphs, as shown by Fig.~\ref{graph:hopreach}.
Therefore when considering \phop for reachability queries processing, its necessary that we can quickly know what is the \rr w.r.t. a set of \hnodes for the underlying graphs, such that we can correctly decide whether we should use \phop, and further, we can decide how many \hnodes should be chosen to construct \phop.

\stitle{Other Reachability Approaches:} Besides approaches that use (partial) 2-hop labels, researchers also proposed other approaches that do not involve 2-hop labels, including~\cite{DBLP:journals/pvldb/YildirimCZ10, DBLP:journals/vldb/YildirimCZ12,DBLP:conf/edbt/VelosoCJZ14, DBLP:conf/icde/SeufertABW13, DBLP:journals/pvldb/WeiYLJ14, DBLP:conf/edbt/ZhangYQZZ12,DBLP:journals/tkde/SuZWY17}.
These approaches assign each \node $u$ a label that maintains partial \tc. For a given reachability query $\query{u}{v}$, we may need to conduct depth-first search (\dfs) or breadth-first search (\bfs) from $u$ to check whether $u$ can reach $v$, if we cannot get the result by comparing labels of $u$ and $v$.

\begin{figure*}[t]
  \centering
\includegraphics{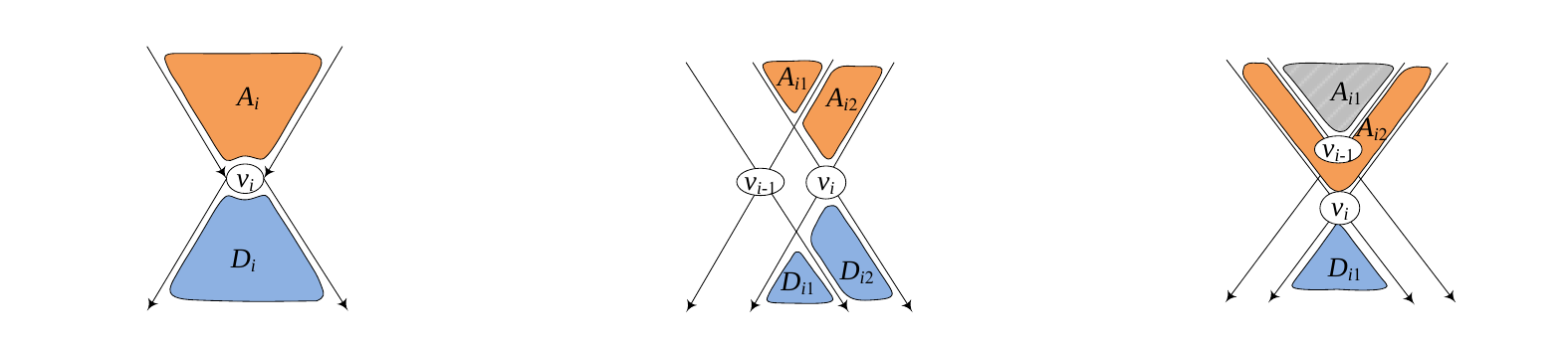}

\flushleft
\quad\quad\qquad\quad\quad(a) $v_i$ is the first node ~\quad\qquad\quad\quad\quad (b) $\queryfalse{v_{i-1}}{v_i}$ and $\queryfalse{v_{i}}{v_{i-1}}$
~~~~\qquad\quad\quad\quad\quad (c) $\querytrue{v_{i-1}}{v_i}$ 

\caption{The relationship between different hop-nodes.}
\label{graph:ourapproach-idea}
\end{figure*}

\section{The Baseline Algorithm}
\label{section:baseline}

To get the \rr of $\Sk{k}$, we need to solve two tasks. One is constructing \phop,
the other is computing the \rr.
In this section, we first analyze the construction of 2-hop labels and the computation of \rr,
then give out the baseline algorithm for \rr computation.

\stitle{Step-1: 2-hop Labels Construction.} To construct 2-hop labels, existing approaches need to sort all nodes based
on a certain rank value, such as degree~\cite{DBLP:journals/pvldb/JinW13, DBLP:conf/cikm/YanoAIY13} or
closeness~\cite{DBLP:conf/sigmod/AkibaIY13}. The result of the sorting operation is $v_1, v_2, ..., \cdots, v_{|V|}$,
where the first (last) node has the largest (smallest) rank value. Based on the sorting result, we select the
first $k$ nodes as \hnodes to
get the \hnode set $\Sk{k}=\{v_1, v_2, ..., v_k\}$. We have the following result w.r.t. the \hnode sets
(Equations~\ref{eq:hnodeset} and \ref{eq:hnodeset-subsumption}).

\begin{equation}\label{eq:hnodeset}
  \emptyset = \Sk{0} \subset \Sk{1} \subset \Sk{2} \subset \cdots \subset \Sk{|V|} = V
\end{equation}

\begin{equation}\label{eq:hnodeset-subsumption}
  \Sk{i}\setminus \Sk{i-1} =\{v_{i}\} \text{, where}~ 0<i\leq |V|
\end{equation}

Given a set $\Sk{i}$ of $i$ \hnodes, its 2-hop labels $\phoplabel{i}$ can be generated by processing $v_i$
based on $\phoplabel{i-1}$
according to Equations~\ref{eq:hnodeset} and \ref{eq:hnodeset-subsumption}.
Specifically, we first perform forward \bfs from $v_i$ to get a set $\dset{i}$ of nodes that $v_i$ can reach.
Second, we perform backward \bfs from $v_i$ to get a set $\aset{i}$ of nodes that can reach $v_i$,
as denoted by Figure~\ref{graph:ourapproach-idea}(a).
We call $\aset{i}$ the ancestor set of $v_i$,
and $\dset{i}$ the descendant set of $v_i$.
For each node $a\in \aset{i}$, we add $v_i$ to $a$'s out-label, i.e., $\lbOut{a}{i}=\lbOut{a}{i-1}\cup\{v_i\}$,
denoting that $a$ can reach
$v_i$. For each node $d\in \dset{i}$, we add $v_i$ to $d$'s in-label, i.e., $\lbIn{d}{i}=\lbIn{d}{i-1}\cup\{v_i\}$,
denoting that $v_i$ can
reach $d$.
After processing $v_i$, we get 2-hop labels $\phoplabel{i}$.
The superscript $i$ in $\lbOut{a}{i}(\lbIn{a}{i})$ denotes that
both 2-hop labels $\lbOut{a}{i}$ and $\lbIn{a}{i}$ w.r.t. node $a$ are
subsets of $\Sk{i}$, i.e., they contain only nodes of $\Sk{i}$. When $i=|V|$, then $\Sk{i}=V$, the 2-hop
labels of $a$ are subsets of $V$. In this case, all nodes are \hnodes and we do not use superscript
in 2-hop labels for simplicity, i.e., $\lbOut{a}{|V|}=\lbOut{a}{}$ and $\lbIn{a}{|V|}=\lbIn{a}{}$.

It is worth noting that we can use 2-hop labels $\phoplabel{i-1}$ to reduce the size of both $\aset{i}$ and $\dset{i}$
by checking whether we can terminate
the \bfs traversal from $v_i$ in advance.
For example, consider Figure~\ref{graph:ourapproach-idea}(c), where both $v_{i-1}$ and $v_i$ are \hnodes
and $v_i$ is processed after $v_{i-1}$.
After processing $v_{i-1}$, we have $\phoplabel{i-1}$.
When processing $v_i$, the backward \bfs traversal from $v_i$ can be terminated at $v_{i-1}$, due to that $v_{i-1}$ can
reach $v_i$ can be answered by $\phoplabel{i-1}$, and $\forall a\in \rtcv{v_{i-1}}$ that
can reach $v_i$ through $v_{i-1}$ can also be answered by $\phoplabel{i-1}$.
Therefore in practice, $\aset{i}\subseteq \InStar{}{v_i}\wedge
\dset{i}\subseteq \OutStar{}{v_i}$.

\begin{figure}
  \centering
\includegraphics{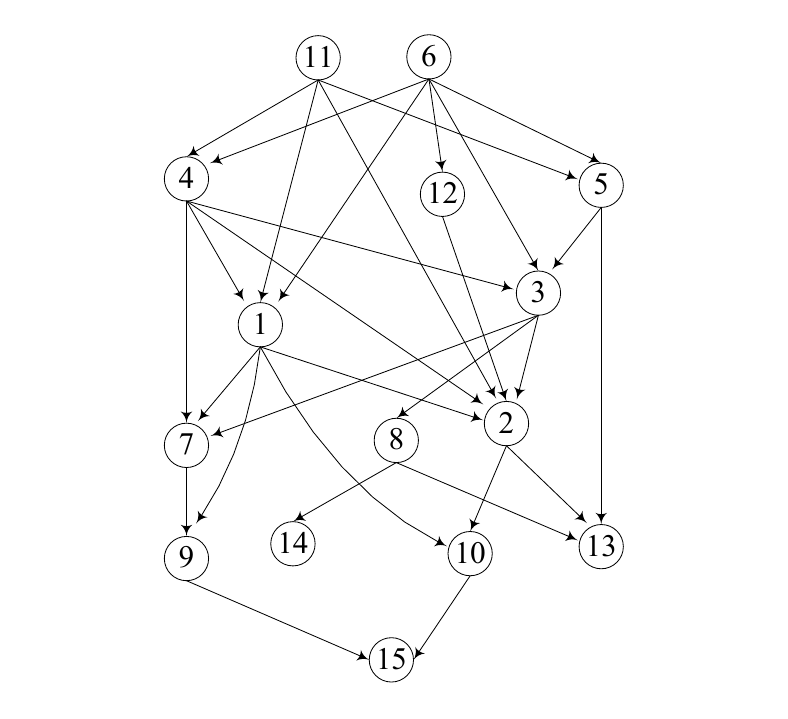}


\caption{A sample \DAG $G$.}
\label{graph:graph}
\end{figure}

\begin{table}[h]
\caption{The \phop $\phoplabel{1},\phoplabel{2},\phoplabel{3}$ constructed based on
$\Sk{1}=\{v_1\},\Sk{2}=\{v_1, v_2\}$ and $\Sk{3}=\{v_1, v_2, v_3\}$, respectively.}
\label{table:graph}

\centering


\begin{tabular} {|c||@{}r|@{}r||@{}r|@{}r||@{}r|@{}r|} \hline

 \raisebox{-2mm}[-10mm][1mm]{Node} & \multicolumn{2}{c||}{$\phoplabel{1}$} & \multicolumn{2}{c||}{$\phoplabel{2}$}&
 \multicolumn{2}{c|}{$\phoplabel{3}$}\\  \cline{2-7}

        & ~$\lbOut{v}{1}$ & ~$\lbIn{v}{1}$   &~$\lbOut{v}{2}$  & ~$\lbIn{v}{2}$ & ~$\lbOut{v}{3}$   &~$\lbIn{v}{3}$   \\

\hline\hline
$v_{1}$	& 	1	& 	1	& 	1	& 	1	& 	1	& 	1	\\	\hline
$v_{2}$	& 		& 	1	& 	2	& 	1,2	& 	2	& 	1,2	\\	\hline
$v_{3}$	& 		& 		& 	2	& 		& 	2,3	& 	3	\\	\hline
$v_{4}$	& 	1	& 		& 	1	& 		& 	1,3	& 		\\	\hline
$v_{5}$	& 		& 		& 	2	& 		& 	2,3	& 		\\	\hline
$v_{6}$	& 	1	& 		& 	1	& 		& 	1,3	& 		\\	\hline
$v_{7}$	& 		& 	1	& 		& 	1	& 		& 	1,3	\\	\hline
$v_{8}$	& 		& 		& 		& 		& 		& 	3	\\	\hline
$v_{9}$	& 		& 	1	& 		& 	1	& 		& 	1,3	\\	\hline
$v_{10}$	& 		& 	1	& 		& 	1,2	& 		& 	1,2	\\	\hline
$v_{11}$	& 	1	& 		& 	1	& 		& 	1,3	& 		\\	\hline
$v_{12}$	& 		& 		& 	2	& 		& 	2	& 		\\	\hline
$v_{13}$	& 		& 	1	& 		& 	1,2	& 		& 	1,2	\\	\hline
$v_{14}$	& 		& 		& 		& 		& 		& 	3	\\	\hline
$v_{15}$	& 		& 	1	& 		& 	1,2	& 		& 	1,2	\\	\hline

\end{tabular}
\end{table}

\vspace{3mm}

\begin{example} \label{example:baseline1} Consider $G$ in Figure~\ref{graph:graph}.
Assume that we want to construct \phop
$\phoplabel{2}$. The first thing we need to do is to sort all nodes by a certain rank
value.
In this paper, we follow the
tradition~\cite{DBLP:conf/cikm/YanoAIY13,DBLP:journals/pvldb/JinW13} and
take $(|\Out{}{v}|+1)\times(|\In{}{v}|+1)$
as $v$'s rank value for sorting. The sorting result is $v_1, v_2, v_3, ..., v_{15}$. To
get $\phoplabel{2}$, we first process $v_1$ by performing both forward and backward
\bfs from $v_1$ to
get $\aset{1}=\{v_1, v_4, v_6, v_{11}\}$ and $\dset{1}=\{v_1, v_2, v_7, v_{9}, v_{10},
v_{13}, v_{15}\}$. After that, we add 1 to out label of nodes in $\aset{1}$
and in label of nodes in $\dset{1}$. Then, we get $\phoplabel{1}$. The next processed
node is $v_2$. Similarly, we perform forward and backward \bfs from $v_2$ to get
$\aset{2}=\{v_2,v_3, v_5, v_{12}\}$ and $\dset{2}=\{v_2, v_{10}, v_{13}, v_{15}\}$. Note that
all nodes in $\aset{1}$ can reach $v_2$, but some of them are not included in $\aset{2}$,
due to that for nodes that are in $\aset{1}$ but not included in $\aset{2}$,
they can reach $v_2$ can be answered by $\phoplabel{1}$. Then, we add 2 to out and in label
of nodes in $\aset{2}$ and $\dset{2}$, as shown by Table~\ref{table:graph}. \hfill ~ $\Box$
\end{example}

\vspace{3mm}

\stitle{Step-2: Reachability Ratio Computation.} Given 2-hop labels $\phoplabel{k}$ w.r.t. $\Sk{k}$,
the baseline approach computes the \rr of $\Sk{k}$ as follow.
First, it computes the set of nodes
that can reach either one of the set of \hnodes, as shown by Equation~\ref{eq:Aunion}.
Second, it computes the set of
nodes that can be reached by either one of the set of \hnodes, as shown by Equation~\ref{eq:Dunion}.
It computes
the number of reachable queries that can be answered by $\phoplabel{k}$, as shown by Equation~\ref{eq:totalcovernumber}.
At last, we return the \rr of 2-hop labels w.r.t. $\Sk{k}$ based on Equation~\ref{eq:reachratio}.

\begin{equation}\label{eq:Aunion}
  A = \bigcup_{i\in [1,k]} \aset{i}
\end{equation}

\begin{equation}\label{eq:Dunion}
  D = \bigcup_{i\in [1,k]} \dset{i}
\end{equation}

\begin{equation}\label{eq:totalcovernumber}
  \begin{split}
    \phopNum{k} = |\{(a,d)|&a\in A, d\in D, a\neq d\, \\
                  & \lbOut{a}{k}\cap \lbIn{d}{k}\neq \emptyset \}|
  \end{split}
\end{equation}

\vspace{3mm}

\begin{example} \label{example:baseline2} Continue Example~\ref{example:baseline1}. To compute the \rr of
$\Sk{2}=\{v_1, v_2\}$, we first compte
$\aset{}=\aset{1}\cup \aset{2} =  \{v_1, v_2,v_3, v_4, v_5, v_6, v_{11}, v_{12}\}$,
$\dset{}=\dset{1}\cup \dset{2} = \{v_1, v_2, v_7, v_{9}, v_{10},
v_{13}, v_{15}\}$ according to Equations~\ref{eq:Aunion} and \ref{eq:Dunion}, respectively.
At last, we check for each pair of nodes $a\in \aset{}$ and $d\in \dset{} (a\neq d)$ , whether
$a$ can reach $d$ can be answered by $\phoplabel{2}$.
And compute the number of answered reachable queries according to
Equation~\ref{eq:totalcovernumber}, which is 42 for $G$ in Figure~\ref{graph:graph} and
$\phoplabel{2}$ in Table~\ref{table:graph}. Given $\tcv{G}=70$, we know that the \rr of $\Sk{2}$
is $42/70=60\%$.
\hfill ~ $\Box$
\end{example}

\vspace{3mm}

\stitle{The Algorithm:} The baseline algorithm to compute \rr is shown in Algorithm~\ref{algorithm:baseline},
which works in two steps. Step-1 (lines 1-17) constructs 2-hop labels $\phoplabel{k}$ of $k$ \hnodes and
gets the two set of nodes $\aset{}$ and $\dset{}$.
Specifically, it first sorts all nodes in certain order in line 2, then selects $k$ \hnodes
in line 3. In lines 4-15, it performs
forward and backward \bfs from each \hnode $v_i$ to construct 2-hop labels $\phoplabel{i}$.
During the processing, only if the reachability
relationship between $v_i$ and the visited node $v$ cannot be answered by 2-hop labels $\phoplabel{i-1}$, it adds
$v_i$ to $v$'s in-label (line 8) or out-label (line 13), and adds $v$ to $\dset{i}$ (line 9)
or $\aset{i}$ (line 14); otherwise, it
terminates the processing due to that the reachability
relationship has already been covered by $\phoplabel{i-1}$.
In lines 16-17, it gets the two sets $\aset{}$ and $\dset{}$ according to Equations~\ref{eq:Aunion} and \ref{eq:Dunion}.
Step-2 (lines 18-20) computes the number of covered reachable queries by $\phoplabel{k}$ according
to Equation~\ref{eq:totalcovernumber}.
Finally, it computes and returns the \rr in line 21.

\begin{algorithm}[t]
 \begin{flushleft}
 \end{flushleft}
\caption{blRR$(G=(V,E), k, \tcv{G})$} \label{algorithm:baseline}

\vspace{-3mm}

~1 \ $\phopNum{k}\leftarrow 0, \Sk{k}\leftarrow \emptyset, \aset{}\leftarrow \emptyset, \dset{}\leftarrow \emptyset$

~2 \ rank nodes in $G$ in a certain order

~3 \ put the first $k$ nodes into $\Sk{k}$ as \hnodes

~4 \ \textbf{foreach} $(v_i\in \Sk{k})$ \textbf{do}

~5 \ \quad \ $\aset{i}\leftarrow \emptyset$; $\dset{i}\leftarrow \emptyset$

~6 \ \quad \ perform forward \bfs from $v_i$, and for each visited $v$

~7 \ \ \quad\quad \textbf{if} $(\lbOut{v_i}{i-1}\bigcap \lbIn{v}{i-1} = \emptyset)$ \textbf{then} \hfill /*$\phoplabel{i-1}$*/

~8 \ \ \quad\quad\quad $\lbIn{v}{i} \leftarrow \lbIn{v}{i-1}\bigcup \{v_{i}\}$ \hfill /*compute $\phoplabel{i}$*/

~9 \ \ \quad\quad\quad $\dset{i} \leftarrow \dset{i} \cup \{v\}$

10 \ \ \quad\quad    \textbf{else} stop expansion from $v$

11 \ \quad perform backward \bfs from $v_i$, and for each visited $v$

12 \ \quad \quad  \textbf{if} $(\lbOut{v}{i-1}\bigcap \lbIn{v_i}{i-1} = \emptyset)$ \textbf{then} \hfill /*$\phoplabel{i-1}$*/

13 \ \quad\quad\quad $\lbOut{v}{i} \leftarrow \lbOut{v}{i-1}\bigcup \{v_i\}$                  \hfill /*compute $\phoplabel{i}$*/

14 \ \ \quad\quad\quad $\aset{i} \leftarrow \aset{i} \cup \{v\}$

15 \ \quad \quad  \textbf{else} stop expansion from $v$

16 \ $A = \bigcup_{i\in [1,k]} \aset{i}$

17 \ $D = \bigcup_{i\in [1,k]} \dset{i}$

18 \ \textbf{foreach} $(a\in \aset{}, d\in \dset{}, a\neq d)$ \textbf{do}

19 \ \ \quad \textbf{if} $(\lbOut{a}{k}\bigcap \lbIn{d}{k} \neq \emptyset)$ \textbf{then}   \hfill /*$\phoplabel{k}$*/

20 \ \ \quad\quad $\phopNum{k} \leftarrow \phopNum{k} + 1$

21 \ \textbf{return} $\alpha \leftarrow \phopNum{k}/\tcv{G}$ as \rr of $\Sk{k}$

\end{algorithm}

\stitle{Analysis:} For Step-1 (lines 1-17), the time cost of line 2 is $O(|V|)$ by counting sort.
The time cost of performing \bfs from each \hnode $v_i$ is $O(|V|+|E|)$ (lines 5-15).
During the two \bfs traversals,
the time cost of processing every visited node $v$ is $O(k)$ (lines 7 and 12). Thus the time cost of 2-hop labels
construction for
each \hnode is $O(k(|V|+|E|))$, and the time cost of processing $k$ \hnodes, i.e., the time cost
of Step-1 is $O(k^2(|V|+|E|))$.
For Step-2, the time cost is $O(k|A| |D|)$.
Therefore, the time complexity of Algorithm~\ref{algorithm:baseline} is $O(k^2(|V|+|E|)+k|A||D|)$.

During the processing, we do not need to actually maintain every $\aset{i}$ and $\dset{i}$, instead,
we only need to maintain $\aset{}$ and $\dset{}$. Further,
we need to maintain the 2-hop labels w.r.t. $k$ \hnodes, the space cost is $O(k|V|)$. As
$\Sk{k}, \aset{}$ and
$\dset{}$ are bounded by $V$, 
the space complexity of Algorithm~\ref{algorithm:baseline} is $O(k|V|)$.

In practice, if the \rr is too small to meet the requirement,
we may need to use more \hnodes, and therefore Algorithm~\ref{algorithm:baseline} will be called once more to compute the new
\rr, for which all reachability relationships tested for $\Sk{k}$ will be tested
again for the new \hnode set.

\vspace{3mm}

\begin{example} \label{example:baseline3} Continue Example~\ref{example:baseline2}.
After getting
$\aset{}=\{v_1, v_2,v_3,$ $ v_4, v_5, v_6, v_{11}, v_{12}\}$ and
$\dset{} = \{v_1, v_2, v_7, v_{9}, v_{10},
v_{13}, v_{15}\}$ during constructing \phop, in lines 18-20, we need to
test 56 reachability queries, due to $|\aset{}|=8$ and $|\dset{}|=7$. By line 21, we know that the \rr is 60\%.
If we set the threshold of
the \rr to be equal or greater than 80\%, then we need to enlarge the \hnode set and
recompute the \rr from scratch. As a result, the 56 queries tested for $\Sk{2}$ will be tested again
for the new \hnode set.
\hfill ~ $\Box$
\end{example}

\vspace{3mm}

\section{The Incremental Approach}
\label{section:incremental}

Considering that when the \hnode set is enlarged, Algorithm~\ref{algorithm:baseline} will be called once more, and the set of reachability relationships tested for the first call will be tested again for the second call, a natural question is: can we compute the \rr incrementally? That is, given the \rr w.r.t. $\Sk{i-1}$, when we decide to compute the \rr w.r.t. $i$ \hnodes, i.e., $\Sk{i}=\Sk{i-1}\cup \{v_{i}\}$, we do not compute the number of covered reachable queries from scratch, instead, we only compute the number of increased reachable queries that cannot be covered by $\phoplabel{i-1}$, but can be covered by $\phoplabel{i}$.

However, the increased \rr cannot be easily computed. On one hand, by constructing 2-hop
labels using \hnode $v_i$, we capture three kinds of
reachability relationships: (1) $v_i$ can reach every node in $\dset{i}\setminus \{v_i\}$ can
be determined by 2-hop labels w.r.t. $v_i$, and the number of covered reachable queries is $|\dset{i}|-1$;
(2) every node in $\aset{i}\setminus \{v_i\}$ can reach $v_i$ can be determined
by 2-hop labels w.r.t. $v_i$, and the number of covered reachable queries is $|\aset{i}|-1$;
and (3) each node in $\aset{i}\setminus \{v_i\}$ can reach every node in $\dset{i}\setminus \{v_i\}$ can be determined by 2-hop labels w.r.t. $v_i$,
and the number of covered reachable queries
is $(|\aset{i}|-1)\times(|\dset{i}|-1)$. Thus the number of covered reachable queries by 2-hop labels w.r.t.
$v_i$ can be computed as $(|\aset{i}|-1)\times(|\dset{i}|-1)+ (|\aset{i}|-1) + (|\dset{i}|-1)=|\aset{i}|\times|\dset{i}|-1$. 

On the other hand, 2-hop labels w.r.t. different \hnodes may cover the same reachable queries.
For example, consider Figure~\ref{graph:ourapproach-idea}(b), where $v_{i-1}$ and $v_i$
are two \hnodes, and $v_i$ is processed after $v_{i-1}$.
After processing $v_{i-1}$, every node $a_1\in \aset{i1}$ can reach every node $d_1\in \dset{i1}$
can be covered by 2-hop labels
w.r.t. $v_{i-1}$, due to that $v_{i-1}\in \lbOut{a_1}{i-1}\cap \lbIn{d_1}{i-1}$.
After processing $v_i$, we also know that $a_1$ can reach
$d_1$ can also be covered by 2-hop labels w.r.t. $v_i$, due to that $v_i\in \lbOut{a_1}{i}\cap \lbIn{d_1}{i}$.
Therefore, the increased number of reachable queries w.r.t. $v_i$ can be computed as Equations~\ref{eq:singlenode} and \ref{eq:singlenode-lambda}, and the total number of reachable queries $\phopNum{k}$ covered by
$\phoplabel{k}$ can be computed as Equation~\ref{eq:singlenode-total}.

\begin{equation}\label{eq:singlenode}
  \phopNk{i} = |\aset{i}|\times|\dset{i}|-1 - \lambda
\end{equation}

\begin{equation}\label{eq:singlenode-lambda}
  \begin{split}
    \lambda = |\{(a,d)|&a\in\aset{i}, d\in\dset{i}, a\neq d, \\
                  & \lbOut{a}{i-1}\cap \lbIn{d}{i-1}\neq \emptyset\}|
  \end{split}
\end{equation}

\begin{equation}\label{eq:singlenode-total}
  \phopNum{k} = \sum_{i\in [1,k]}\phopNk{i}
\end{equation}

\vspace{3mm}

Therefore, to compute the number of reachable queries that cannot be covered by $\phoplabel{i-1}$ but can be covered by $\phoplabel{i}$,
the intuitive way is firstly getting the two sets of nodes $\aset{i}$ and $\dset{i}$, then testing for each pair of nodes $a\in \aset{i}$ and $d\in \dset{i}$, whether $a$ can reach $d$ can be answered by $\phoplabel{i-1}$.
If $a$ can reach $d$ can be answered by $\phoplabel{i-1}$, it means that $\querytrue{a}{d}$ has already been covered by $\Sk{i-1}$;
otherwise, it is a new covered reachable query and needs to be counted in, as shown by Algorithm~\ref{algorithm:icrRR}.

\begin{algorithm}[t]
 \begin{flushleft}
 \end{flushleft}
\caption{incRR$(G=(V,E), k, \tcv{G})$} \label{algorithm:icrRR}

\vspace{-3mm}


~1 \ \ $\phopNum{0}\leftarrow 0$

~2 \ \ rank nodes in $G$ in a certain order

~3 \ \ put the first $k$ nodes into $\Sk{k}$ as \hnodes

~4 \ \ \textbf{foreach} $(v_i\in \Sk{k})$ \textbf{do}

~5 \ \ \quad \ $\aset{i}\leftarrow \emptyset, \dset{i}\leftarrow \emptyset$

~6 \ \ \quad \ perform forward \bfs from $v_i$, and for each visited $v$

~7 \ \ \ \quad\quad \textbf{if} $(\lbOut{v_i}{i-1}\bigcap \lbIn{v}{i-1} = \emptyset)$ \textbf{then} \hfill /*$\phoplabel{i-1}$*/


~8 \ \ \quad\quad\quad $\dset{i} \leftarrow \dset{i} \cup \{v\}$

~9   \ \ \quad\quad    \textbf{else} stop expansion from $v$

10 \ \quad perform backward \bfs from $v_i$, and for each visited $v$

11 \ \quad \quad  \textbf{if} $(\lbOut{v}{i-1}\bigcap \lbIn{v_i}{i-1} = \emptyset)$ \textbf{then} \hfill /*$\phoplabel{i-1}$*/


12 \ \ \quad\quad\quad $\aset{i} \leftarrow \aset{i} \cup \{v\}$

13 \ \quad \quad  \textbf{else} stop expansion from $v$

14 \ \ \quad $\lambda\leftarrow 0$

15 \ \ \quad \textbf{foreach} $(a\in \aset{i}, d\in \dset{i}, a\neq d)$ \textbf{do}

16 \ \ \quad\quad    \textbf{if} $(\lbOut{a}{i-1}\bigcap \lbIn{d}{i-1} \neq \emptyset)$ \textbf{then}   \hfill /*$\phoplabel{i-1}$*/

17 \ \ \quad\quad\quad $\lambda \leftarrow \lambda + 1$

18 \ \ \quad $\phopNk{i}\leftarrow |\aset{i}|\times|\dset{i}|-1-\lambda$

19 \ \ \quad $\phopNum{i} \leftarrow \phopNum{i-1} + \phopNk{i}$

20 \ \ \quad $\alpha \leftarrow \phopNum{i}/\tcv{G}$ \hfill /*\rr of $\Sk{i}$*/

21 \ \ \quad \textbf{foreach} $(a\in \aset{i})$ \textbf{do} \hfill /*compute $\phoplabel{i}$*/

22 \ \ \quad\quad\quad $\lbOut{a}{i} \leftarrow \lbOut{a}{i-1}\bigcup \{v_i\}$

23 \ \ \quad \textbf{foreach} $(d\in \dset{i})$ \textbf{do}\hfill /*compute $\phoplabel{i}$*/

24 \ \ \quad\quad\quad $\lbIn{d}{i} \leftarrow \lbIn{d}{i-1}\bigcup \{v_{i}\}$

25 \ \textbf{return} $\alpha$ as \rr of $\Sk{k}$

\end{algorithm}

In Algorithm~\ref{algorithm:icrRR}, we compute \rr for each $\Sk{i}$ when $v_i(i\in [1,k])$ is added into $\Sk{i-1}$. For each processed $v_i$ (lines 4-24), we first perform forward and backward \bfs from $v_i$ to get the two set of nodes $\aset{i}$ and $\dset{i}$ (lines 5-13). In lines 14-17, we compute the number of reachable queries that can be covered by $\phoplabel{i-1}$.
After that, we get the increased number of reachable queries that can be covered by $\phoplabel{i}$ but cannot be covered by $\phoplabel{i-1}$ in line 18 according to Equation~\ref{eq:singlenode}. In line 19, we get the total number of reachable queries covered by $\phoplabel{i}$, and get the \rr of $\Sk{i}$ in line 20.
We compute $\phoplabel{i}$ based on $\phoplabel{i-1}$ in lines 21-24. At last, we return the \rr of $\Sk{k}$ in line 25.

It is worth noting that for Algorithm~\ref{algorithm:icrRR}, when processing the first \hnode $v_1$, we do not need to actually execute lines 15-17, due to that for $v_1$, $\aset{1}=\InStar{}{v_1}$ and $\dset{1}=\OutStar{}{v_1}$, we can directly get $\phopNk{i}=|\aset{i}|\times|\dset{i}|-1$ and the corresponding \rr.

\stitle{Analysis:} Different with Algorithm~\ref{algorithm:baseline}, Algorithm~\ref{algorithm:icrRR} performs Step-1 by first computing the two sets $\aset{i}$ and $\dset{i}$ in lines 1-13, then computing $\phoplabel{i}$ in lines 21-24. The overall cost is same as that of Algorithm~\ref{algorithm:baseline}, i.e., $O(k^2(|V|+|E|))$.

The difference between Algorithm~\ref{algorithm:baseline} and Algorithm~\ref{algorithm:icrRR} lies in Step-2 (lines 14-20), i.e., how to compute the increased number of reachable queries that cannot be covered by $\phoplabel{i-1}$ but can be covered by $\phoplabel{i}$ based on Equation~\ref{eq:singlenode}. The cost of Step-2 for each \hnode is $O(i|\aset{i}||\dset{i}|)$. For $k$ \hnode, the cost is therefore $O(\sum_{i\in [1,k]}i|\aset{i}||\dset{i}|)$.

Therefore, the time complexity of Algorithm~\ref{algorithm:icrRR} is $O(k^2(|V|+|E|)+\sum_{i\in [1,k]}i|\aset{i}||\dset{i}|)$.

Similar to Algorithm~\ref{algorithm:baseline}, we need to maintain the 2-hop labels w.r.t. at most $k$ \hnodes during the processing. As
$\aset{i}$ and
$\dset{i}$ are bounded by $V$, and $\aset{i}(\dset{i})$ can be used to store nodes of $\aset{i+1}(\dset{i+1})$,
the space complexity of Algorithm~\ref{algorithm:icrRR} is $O(k|V|)$.

\vspace{3mm}

\begin{example} \label{example:incremental1} Consider $G$ in Figure~\ref{graph:graph}. Assume that we want to construct \phop $\phoplabel{3}$.

The first node to be processed is $v_1$, and
the \phop are shown in Table~\ref{table:graph}. As $\aset{1}=\{v_1, v_4, v_6, v_{11}\}$, $\dset{1}=\{v_1,v_2,v_7, v_9, v_{10},v_{13}, v_{15}\}$, thus we know that $\phopNum{1}=\phopNk{1}=|\aset{1}|\times|\dset{1}|-1=27$.
The second processed node is $v_2$ and the
\phop are shown in Table~\ref{table:graph}. By lines 5-13, we have that $\aset{2}=\{v_2, v_3, v_5, v_{12}\}$, $\dset{2}=\{v_2, v_{10}, v_{13}, v_{14}\}$. Then, in lines 15-17, we need to test $|\aset{2}|\times|\dset{2}|=16$ reachability queries. The result is $\lambda=0$, thus $\phopNk{2}=|\aset{2}|\times|\dset{2}|-1-0=15$, and $\phopNum{2}=\phopNk{1}+\phopNk{2}=27+15=42$.
The third processed node is $v_3$. By lines 5-13, we have that $\aset{3}=\{v_3, v_4, v_5, v_6, v_{11}\}$, $\dset{3}=\{v_3,v_7, v_8, v_{9}, v_{14}\}$. Then, in lines 15-17, we need to test $|\aset{3}|\times|\dset{3}|=25$ reachability queries. The result is $\lambda=6$, thus $\phopNk{3}=|\aset{3}|\times|\dset{3}|-1-6=18$, and $\phopNum{3}=\phopNum{2}+\phopNk{3}=42+18=60$.
After processing $v_1, v_2$ and $v_3$, we have $\phoplabel{3}$ shown in Table~\ref{table:graph}, and the \rr is $60/70=85.7$\% by testing $16+25=41$ reachability queries for Algorithm~\ref{algorithm:icrRR}.

As a comparison, when using Algorithm~\ref{algorithm:baseline}, $|\aset{3}|=8$, $|\dset{3}|=10$, and we need to test 80 reachability queries to get the \rr.
\hfill ~ $\Box$
\end{example}

\vspace{3mm}

Note that, since a reachable query $\querytrue{a}{d}$ may be covered by 2-hop labels w.r.t. different hop-nodes, compared with Algorithm~\ref{algorithm:baseline}, Algorithm~\ref{algorithm:icrRR} may test the reachability relationship between $a$ and $d$ in line 16 more than once. However, it is still valuable due to that (1) only a part of reachable queries, rather than all, need to be tested more than once, and (2) we can terminate the computation whenever we find that \rr of $\Sk{i}$ is good enough in line 20. As a comparison, Algorithm~\ref{algorithm:baseline} may be called more than once before getting a \rr meeting the requirement. When it is called again due to the enlarged \hnode set, all previously tested queries will be tested once more.

\section{The Incremental-Partition Approach}
\label{section:optimization}

\subsection{The Equivalence Relationship}

By comparing Algorithm~\ref{algorithm:baseline} and Algorithm~\ref{algorithm:icrRR}, we know that the key factor that affects the overall performance is the total number of tested reachability queries, which dominates the cost of Step-2, as indicated by their time complexities. Even though Algorithm~\ref{algorithm:icrRR} does not need to compute \rr from scratch when the \hnode set becomes large by adding one more \hnode $v_i$, it still needs to test $|\aset{i}|\times |\dset{i}|$ reachability queries in line 16 with cost $O(i|\aset{i}||\dset{i}|)$. Given a large \hnode set, the cost could be unaffordable.

\vspace{3mm}

\begin{definition}\label{def:equivalent} \textbf{[Equivalence Relationship]} Given a \hnode $v_i$, its ancestor set $\aset{i}$ and descendant set $\dset{i}$. We say two nodes $a_1, a_2 (a_1\neq a_2)$ of $\aset{i}$ are forward equivalent to each other, denoted as $a_1 \equiv_F a_2$, if they have the same out-label, i.e., $\lbOut{a_1}{i} = \lbOut{a_2}{i}$. Similarly, we say two nodes $d_1, d_2 (d_1\neq d_2)$ of $\dset{i}$ are backward equivalent to each other, denoted as $d_1 \equiv_B d_2$, if they have the same in-label, i.e., $\lbIn{d_1}{i} = \lbIn{d_2}{i}$.
\end{definition}

\vspace{3mm}

By Definition~\ref{def:equivalent}, we can get, for $\aset{i}$, a partition $\parta{i}=\{\aset{i1}, \aset{i2}, ..., \aset{im}\}$, which consists of a set of $m$ disjoint subsets satisfying that (1) $\forall l,j\in [1,m], l\neq j, \aset{il} \cap \aset{ij} = \emptyset$ and $\cup_{l\in [1,m]}\aset{il} = \aset{i}$; and (2) $\forall a_{l}, a_{j}$ belonging to the same subset, $a_{l}\equiv_F a_{j}$.

\vspace{3mm}

\begin{theorem} \label{theorem:fwdequivalent} Let $a_1$ and $a_2$ be two nodes satisfying that $a_1$ and $a_2$ are forward equivalent to each other $(a_1\equiv_F a_2)$. For $\forall d\in V$, we have that $\lbOut{a_1}{i}\cap \lbIn{d}{i} = \lbOut{a_2}{i}\cap \lbIn{d}{i}$.
\end{theorem}

\begin{proof} The correctness is obvious, due to that $a_1$ and $a_2$ are forward equivalent to each other, which means that they have the same out-label. \hfill $\Box$
\end{proof}

\vspace{3mm}

Based on this result, for each subset $\aset{ij}\in \parta{i}$, to know the reachability relationship from all nodes of $\aset{ij}$ to $\forall d\in V$, we do not need to test $|\aset{ij}|$ reachability queries, instead, we only need to test one reachability query, due to that all nodes of $\aset{ij}$ are forward equivalent to each other.

For $\dset{i}$, we also have a partition $\partd{i}=\{\dset{i1}, \dset{i2}, ..., \dset{in}\}$ satisfying that (1) $\forall l,j\in [1,n], l\neq j, \dset{il} \cap \dset{ij} = \emptyset$ and $\cup_{j\in [1,n]}\dset{ij} = \dset{i}$;
and (2) $\forall d_{l}, d_{j}$ belonging to the same subset, $d_{l}\equiv_B d_{j}$. And similarly, for each subset $\dset{ij}\in \partd{i}$, to know the reachability relationship from any node to all nodes of $\dset{ij}$, we do not need to test $|\dset{ij}|$ reachability queries, instead, the number of tested reachability queries can be reduced to one, due to that all nodes of $\dset{ij}$ are backward equivalent to each other.

\vspace{3mm}

\begin{theorem}\label{theorem:numberofequivalent}
Given a \hnode $v_i$, its ancestor set $\aset{i}$ and descendant set $\dset{i}$, the number of tested reachability queries for \rr computation is $|\parta{i}|\times|\partd{i}|$, which is bounded by $\min\{|\aset{i}|\times |\dset{i}|, 4^{i-1}\}$.

\end{theorem}

\begin{proof} Let $\parta{i}(\partd{i})$ be the partition of $\aset{i}(\dset{i})$ based on the equivalence relationship, $\partVA{i}(\partVD{i})$ the partition of $V$ w.r.t. \hnode set $\Sk{i}$ and forward (backward) equivalence relationship, i.e., all nodes in each subset have the same out-label (in-label), which is a subset of $\Sk{i}$. Initially, $\parta{1}=\{\aset{1}\}(\partd{1}=\{\dset{1}\})$, $\partVA{1}=\{\aset{1}, V\setminus \aset{1}\}(\partVD{1}=\{\dset{1}, V\setminus \dset{1}\})$.

On one hand, according to Theorem~\ref{theorem:fwdequivalent}, for each subset $\aset{il}\in\parta{i}$, the result of testing all the reachability relationships from nodes of $\aset{il}$ to any other node is same to each other, thus we only need to randomly pick a node and take it as the representative node of $\aset{il}$ to perform the testing of reachability relationship.
Similarly, for each subset $\dset{ij}\in\partd{i}$ based on backward equivalence relationship, we can also randomly pick a node and take it as the representative node of $\dset{ij}$ to test the reachability relationships from any node to all nodes of $\dset{ij}$.
As a result, the number of tested reachability queries from nodes of $\aset{i}$ to nodes of $\dset{i}$ is $|\parta{i}|\times|\partd{i}|$. Since $\parta{i}(\partd{i})$ is the partition of $\aset{i}(\dset{i})$, we know that $|\parta{i}|\times|\partd{i}|\leq |\aset{i}|\times|\dset{i}|$.

On the other hand, given the partition $\partVA{i-1}(\partVD{i-1})$ of $V$, the size of $\partVA{i}(\partVD{i})$ is at most twice bigger than that of $\partVA{i-1}(\partVD{i-1})$. The reason lies in that all nodes in each subset of $\partVA{i-1}(\partVD{i-1})$ can be further divided into at most two disjoint subsets. One consists of nodes that can reach (be reached by) $v_i$, and the other contains nodes that cannot reach (be reached by) $v_i$. Then, the size of $\partVA{i}(\partVD{i})$ is bounded by $2^{i}$, and the size of $\parta{i}(\partd{i})$ is bounded by $2^{i-1}$, thus the the number of tested reachability queries is bounded by $2^{i-1}\times 2^{i-1}=4^{i-1}$.

In summary, we know that the number of tested reachability queries for \rr computation w.r.t. \hnode $v_i$ is bounded by $\min\{|\aset{i}|\times |\dset{i}|, 4^{i-1}\}$.
\hfill $\Box$

\end{proof}

\vspace{3mm}

According to Theorem~\ref{theorem:numberofequivalent}, we can reduce the number of tested reachability queries when processing \hnode $v_i$.

As shown by Equation~\ref{eq:singlenode-partition}, for each pair of subsets $(\aset{il}\in \parta{i}, \dset{ij}\in \partd{i})$, we only need to check the reachability relationship between their representative nodes $a\in \aset{il}$ and $d\in \dset{ij}$. If $a$ can reach $d$ can be answered by $\phoplabel{i-1}$, it means that all the reachable relationships from each node of $\aset{il}$ to every node of $\dset{ij}$ can be answered by $\phoplabel{i-1}$. To do that, the first thing we need to do is getting the partitions of both $\aset{i}$ and $\dset{i}$ according to equivalence relationship.

\begin{equation}\label{eq:singlenode-partition}
  \lambda = \sum_{
  \substack{a\in \aset{il}\in \parta{i}\\
            d\in \dset{ij}\in \partd{i}\\
            \lbOut{a}{i-1}\cap \lbIn{d}{i-1}\neq \emptyset}}
  |\aset{il}|\times |\dset{ij}|
\end{equation}

\subsection{Partitions Computation}

To get the partitions of both $\aset{i}$ and $\dset{i}$, we need to compare the labels of nodes in $\aset{i}$ and $\dset{i}$, such that nodes with same labels can be clustered together. The naive way to do this is based on pairwise comparing node labels, which is expensive in practice.
It is worth noting that the out-label and in-label of a node are sorted in advance, this actually can be done without additional cost, due to that these labels are used to check whether their set-intersection is empty, which means that we can store the processing order of \hnodes, rather than their IDs, in these labels. In this way, the integers in both out-label and in-label of any node are naturally sorted, as shown by Table~\ref{table:graph}.
With this result, we can sort all nodes in $\aset{i}(\dset{i})$ by comparing their out-labels (in-labels) in lexicographic order. After the sorting operation, all equivalent nodes are clustered together. As the size of each label is bounded by $i$, the cost of computing the partition $\parta{i}(\partd{i})$ of $\aset{i}(\dset{i})$ is $O(i\times|\aset{i}|\times \log |\aset{i}|) (O(i\times|\dset{i}|\times \log |\dset{i}|))$.

Let $\partVA{i}$ be the partition of $V$ w.r.t. \hnode set $\Sk{i}$ and forward equivalence relationship, i.e., all nodes in each subset have the same out-label, which is a subset of $\Sk{i}$, $\partVD{i}$ the partition of $V$ w.r.t. $\Sk{i}$ and backward equivalence relationship, i.e., all nodes in each subset have the same in-label, which is also a subset of $\Sk{i}$. We have the following result.

\vspace{3mm}

\begin{theorem} \label{theorem:partition} Given the \hnode $v_i$ and its ancestor (descendant) set $\aset{i}(\dset{i})$, for $\forall v_1, v_2\in \aset{i}(\dset{i})$, $v_1 \equiv_F v_2 (v_1 \equiv_B v_2)$, iff they belong to the same subset of $\partVA{i-1}(\partVD{i-1})$.
\end{theorem}

\begin{proof} We prove this result from two aspects. First, we prove the correctness when both $v_1$ and $v_2$ belong to $\aset{i}$ (Case-1), then we prove the correctness when both $v_1$ and $v_2$ belong to $\dset{i}$ (Case-2).

\stitle{Case-1} where $v_1, v_2\in \aset{i}$ and $v_i\in \lbOut{v_1}{i}\cap \lbOut{v_2}{i}$.

On one hand, if $v_1 \equiv_F v_2$, it means that $\lbOut{v_1}{i}=\lbOut{v_2}{i}$ according to Definition~\ref{def:equivalent}. Hence, $\lbOut{v_1}{i}\setminus \{v_i\} = \lbOut{v_2}{i}\setminus \{v_i\}$, i.e., they belong to the same subset of $\partVA{i-1}$.

On the other hand, if both $v_1$ and $v_2$ belong to the same subset of $\partVA{i-1}$, it means that before processing \hnode $v_i$, $\lbOut{v_1}{i-1}=\lbOut{v_2}{i-1}$ according to the definition of $\partVA{i-1}$. As $v_1, v_2\in \aset{i}$, we know that after processing $v_i$, $v_i\in \lbOut{v_1}{i}\cap \lbOut{v_2}{i}$ and $\lbOut{v_1}{i} = \lbOut{v_2}{i}$ still holds. According to Definition~\ref{def:equivalent}, $v_1 \equiv_F v_2$.

Therefore we have that $v_1 \equiv_F v_2$, iff they belong to the same subset of $\partVA{i-1}$.

\stitle{Case-2} where $v_1, v_2\in \dset{i}$ and $v_i\in \lbIn{v_1}{i}\cap \lbIn{v_2}{i}$.

Similar to the proof of Case-1, we know that $v_1 \equiv_B v_2$, iff they belong to the same subset of $\partVD{i-1}$.

By considering both the two cases, we know that for $\forall v_1, v_2\in \aset{i}(\dset{i})$, $v_1 \equiv_F v_2 (v_1 \equiv_B v_2)$, iff they belong to the same subset of $\partVA{i-1}(\partVD{i-1})$.
\hfill $\Box$
\end{proof}

\vspace{3mm}

According to Theorem~\ref{theorem:partition}, we assign each node $v$ two set IDs, denoted as $\sid{v}$ and $\sidd{v}$, which are used to check which subset it belongs to in $\partVA{i}$ and $\partVD{i}$, respectively.
Then, given the ancestor (descendant) set $\aset{i}(\dset{i})$ of $v_i$, we only need to scan all nodes of $\aset{i}(\dset{i})$ once, and know immediately that for two nodes $v_1$ and $v_2$, if $\sid{v_1}=\sid{v_2}(\sidd{v_1}=\sidd{v_2})$ in $\partVA{i-1}(\partVD{i-1})$, then $v_1\equiv_F v_2 (v_1\equiv_B v_2)$ and will definitely belong to the same subset of $\partVA{i}(\partVD{i})$. Therefore, $\partVA{i}(\partVD{i})$ is a refinement of $\partVA{i-1}(\partVD{i-1})$, i.e., each element of $\partVA{i}(\partVD{i})$ is a subset of a unique element of $\partVA{i-1}(\partVD{i-1})$.

Recall that when processing the \hnode $v_i$, we first have its ancestor (descendant) set $\aset{i}(\dset{i})$, then get the partition $\parta{i}(\partd{i})$ of $\aset{i}(\dset{i})$ based on equivalence relationship.
Since $\partVA{i}(\partVD{i})$ is the partition of $V$ w.r.t. equivalence relationship, we know that $\parta{i}\subset \partVA{i}(\partd{i}\subset \partVD{i})$, and the relationship between $\partVA{i}(\partVD{i})$, $\partVA{i-1}(\partVD{i-1})$ and $\parta{i}(\partd{i})$ are shown as Equations~\ref{global-partition-relationship-A}-\ref{global-partition-relationship-D}.

\begin{equation}\label{global-partition-relationship-A}
  \partVA{i} = \{P\setminus \aset{i}| P \in \partVA{i-1}\} \cup \parta{i}
\end{equation}

\begin{equation}\label{global-partition-relationship-D}
  \partVD{i} = \{P\setminus \dset{i}| P \in \partVD{i-1}\} \cup \partd{i}
\end{equation}

When processing \hnode $v_i$, since we only need to check the reachability relationships from nodes of $\aset{i}$ to $\dset{i}$, we choose to maintain the information of both $\partVA{i}(\partVD{i})$ and $\parta{i}(\partd{i})$ using the set ID of each node to facilitate partitions computation.
Specifically, we use a hash table $\ahash{}(\dhash{})$ to help achieve linear-time complexity. Each element of $\ahash{}(\dhash{})$ is a tuple $(id_o, e_n)$ denoting a subset $\aset{il}(\dset{il})$ of $\parta{i}(\partd{i})$, where $id_o$ is, for all nodes of $\aset{il}(\dset{il})$, their old set ID in $\partVA{i-1}(\partVD{i-1})$, $e_n=(id_n, v_1, s)$ is a triple denoting the new set ID for all nodes of $\aset{il}(\dset{il})$, the representative node of $\aset{il}(\dset{il})$, and the size of $\aset{il}(\dset{il})$, respectively.

\vspace{3mm}

\begin{example} \label{example:incremental2} Consider $G$ in Figure~\ref{graph:graph}. Before processing $v_1$, $\partVA{0}=\partVD{0}=\{V\}$, $\parta{0}=\partd{0}=\emptyset$, and for all nodes $v$, $\sid{v}=\sidd{v}=0$.

For the first node $v_1$, $\aset{1}=\{v_1, v_4, v_6, v_{11}\}$, $\dset{1}=\{v_1,v_2,v_7, v_9, v_{10},v_{13}, v_{15}\}$. Since all nodes in $\aset{1} (\dset{1})$ have the same $\sid{v}(\sidd{v})$, we know that
$\parta{1}=\{\aset{1}\}$ and $\partd{1}=\{\dset{1}\}$.
$\partVA{1}=\{\aset{1}, V\setminus \aset{1}\}$ and $\partVD{1}=\{\dset{1}, V\setminus \dset{1}\}$. In Table~\ref{table:sid}, the two columns under $v_1$ denote $\partVA{1}$ and $\partVD{1}$, where each 1 in the second (third) column corresponds a node in $\aset{1} (\dset{1})$. Figure~\ref{graph:hash}(a) shows the two hash tables denoting $\parta{1}$ and $\partd{1}$, respectively. For $\ahash{}$, there is one (key, value) pair, denoting that $\parta{1}$ contains one subset $\aset{1}$, and for all nodes in $\aset{1}$, their set ID is 0 in $\partVA{0}$, thus they all belong to the same subset in $\parta{1}$, i.e.,  $\parta{1}=\{\aset{1}\}$. By $\ahash{}$ in Figure~\ref{graph:hash}(a), we know that all nodes in $\aset{1}$ now have the new set ID 1, the representative node of $\aset{1}$ is $v_4$, and $|\aset{1}|=4$.

For the second processed node $v_2$, $\aset{2}=\{v_2, v_3, v_5, v_{12}\}$, $\dset{2}=\{v_2, v_{10}, v_{13}, v_{14}\}$. As all nodes in $\aset{2}$ have the same set ID 0 in $\partVA{1}$, $\parta{2}$ contains a unique subset $\aset{2}$, i.e., $\parta{2}=\{\aset{2}\}$. As shown by Figure~\ref{graph:hash}(b), the key is 0, and the triple $(2, v_3,4)$ denotes that the new set ID for all nodes in $\aset{2}$ is 2, the representative node of $\aset{2}$ is $v_3$, and $|\aset{2}|=4$. Similarly, all nodes in $\dset{2}$ have the same set ID 1 in $\partVD{1}$, $\partd{2}$ contains a unique subset $\dset{2}$, i.e., $\partd{2}=\{\dset{2}\}$, which is denoted as $\dhash{}$ in Figure~\ref{graph:hash}(b).

For the third processed node $v_3$, $\aset{3}=\{v_3, v_4, v_5, v_6, v_{11}\}$, $\dset{3}=\{v_3,v_7, v_8, v_{9}, v_{14}\}$. For $\aset{3}$, $v_3$ and $v_5$ have the same set ID 2 in $\partVA{2}$, thus they form the subset in $\parta{3}$. Further, $v_4, v_6, v_{11}$ have the same set ID 1 in $\partVA{2}$, they form the second subset in $\parta{3}$. Therefore $\parta{3}=\{\{v_3, v_5\},\{v_4, v_6, v_{11}\}\}$. Similarly, we know that $\partd{3}=\{\{v_3,v_8, v_{14}\},\{v_7, v_{9}\}\}$. Both $\parta{3}$ and $\partd{3}$ are denoted by $\ahash{}$ and $\dhash{}$ in Figure~\ref{graph:hash}(c), respectively.
\hfill ~ $\Box$
\end{example}

\vspace{3mm}

\begin{table}[thp]
\caption{The status of set IDs for all nodes.}
\label{table:sid}

\centering


\begin{tabular} {|c||@{}r|@{}r||@{}r|@{}r||@{}r|@{}r|} \hline

 \raisebox{-2mm}[-10mm][1mm]{~Node~} & \multicolumn{2}{c||}{$v_1$} & \multicolumn{2}{c||}{$v_{2}$}&
 \multicolumn{2}{c|}{$v_{3}$}\\  \cline{2-7}

        & ~$\sid{v}$ & ~$\sidd{v}$   &~$\sid{v}$ & ~$\sidd{v}$ & ~$\sid{v}$ & ~$\sidd{v}$   \\

\hline\hline
$v_{1}$	& 	1	& 	1	& 	1	& 	1	& 	1	& 	1	\\	\hline
$v_{2}$	& 		& 	1	& 	2	& 	2	& 	2	& 	2	\\	\hline
$v_{3}$	& 		& 		& 	2	& 		& 	3	& 	3	\\	\hline
$v_{4}$	& 	1	& 		& 	1	& 		& 	4	& 		\\	\hline
$v_{5}$	& 		& 		& 	2	& 		& 	3	& 		\\	\hline
$v_{6}$	& 	1	& 		& 	1	& 		& 	4	& 		\\	\hline
$v_{7}$	& 		& 	1	& 		& 	1	& 		& 	4	\\	\hline
$v_{8}$	& 		& 		& 		& 		& 		& 	3	\\	\hline
$v_{9}$	& 		& 	1	& 		& 	1	& 		& 	4	\\	\hline
$v_{10}$	& 		& 	1	& 		& 	2	& 		& 	2	\\	\hline
$v_{11}$	& 	1	& 		& 	1	& 		& 	4	& 		\\	\hline
$v_{12}$	& 		& 		& 	2	& 		& 	2	& 		\\	\hline
$v_{13}$	& 		& 	1	& 		& 	2	& 		& 	2	\\	\hline
$v_{14}$	& 		& 		& 		& 		& 		& 	3	\\	\hline
$v_{15}$	& 		& 	1	& 		& 	2	& 		& 	2	\\	\hline

\end{tabular}
\end{table}

\begin{figure}[t]
  \centering
\includegraphics{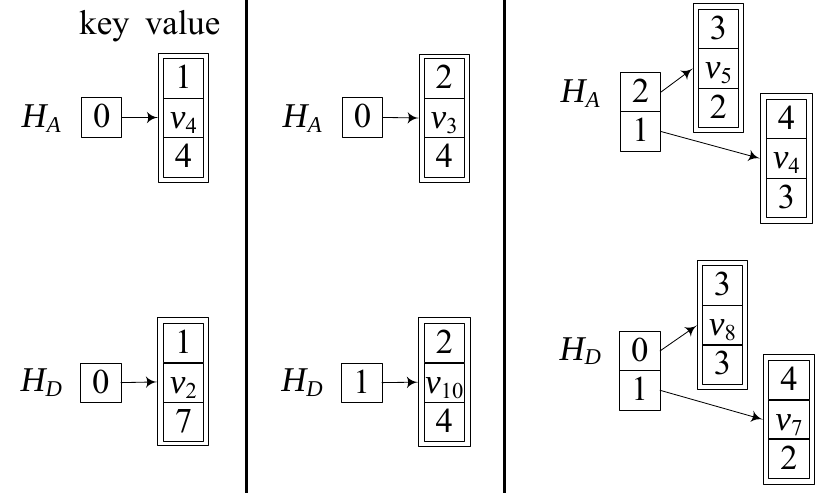}

\flushleft
\quad(a) $\Sk{1}=\{v_1\}$~\quad (b) $\Sk{2}=\{v_1,v_2\}$
~~~~\quad (c) $\Sk{3}=\{v_1,v_2£¬v_3\}$

\caption{Running status of the two hash tables $\ahash{}$ and $\dhash{}$.}
\label{graph:hash}
\end{figure}

\stitle{The Algorithm:} As shown by Algorithm~\ref{algorithm:icrRRplus}, for each \hnode $v_i$, we first perform forward and backward \bfs to get $\dset{i}$ (lines 6-15) and $\aset{i}$ (lines 16-25). At the same time, we generate their partitions $\partd{i}$ and $\parta{i}$, for which each subset is recorded in $\dhash{}$ and $\ahash{}$, respectively. In lines 26-29, we compute $\lambda$ according to Equation~\ref{eq:singlenode-partition}, which is the number of reachable queries that are covered by $\phoplabel{i-1}$. In line 30, we get the number of reachable queries that can be covered by $\phoplabel{i}$ but cannot be covered by $\phoplabel{i-1}$. After that, we have the total number of covered reachable queries in line 31, and the \rr in line 32. At last, we generate $\phoplabel{i}$ in lines 33-36, and return the \rr of $\Sk{k}$ in line 37.

\begin{algorithm}[t]
 \begin{flushleft}
 \end{flushleft}
\caption{incRR$^+(G=(V,E), k, \tcv{G})$} \label{algorithm:icrRRplus}

\vspace{-3mm}

~1 \ \ $\phopNum{0}\leftarrow 0$; $n_A\leftarrow 0$; $n_D\leftarrow 0$;

~2 \ \ rank nodes in $G$ in a certain order

~3 \ \ put the first $k$ nodes into $\Sk{k}$ as \hnodes

~4 \ \ \textbf{foreach} $(v_i\in \Sk{k})$ \textbf{do}

~5 \ \ \quad \ $\aset{i}\leftarrow \emptyset$; $\dset{i}\leftarrow \emptyset$; $\ahash{}\leftarrow \emptyset$; $\dhash{}\leftarrow \emptyset$

~6 \ \ \quad \ perform forward \bfs from $v_i$, and for each visited $v$

~7 \ \ \ \quad\quad \textbf{if} $(\lbOut{v_i}{i-1}\bigcap \lbIn{v}{i-1} = \emptyset)$ \textbf{then} \hfill /*$\phoplabel{i-1}$*/

~8 \ \ \ \  \quad\quad\quad \textbf{if} $(\sid{v}\not\in \dhash{})$ \textbf{then} 

~9 \ \ \ \ \ \quad\quad\quad\quad $n_{A}\leftarrow n_{A}+1$

10 \ \ \quad\quad\quad\quad $\dhash{}[\sid{v}]\leftarrow (n_{A}, v, 1)$

11 \   \quad\quad\quad \textbf{else}

12 \ \ \quad\quad\quad\quad $\dhash{}[\sid{v}].s \leftarrow \dhash{}[\sid{v}].s+1$

13 \  \quad\quad\quad $\dset{i} \leftarrow \dset{i} \cup \{v\}$

14 \  \quad\quad\quad $\sid{v} \leftarrow n_{A}$

15 \  \quad\quad    \textbf{else} stop expansion from $v$

16 \ \quad perform backward \bfs from $v_i$, and for each visited $v$

17 \ \quad \quad  \textbf{if} $(\lbOut{v}{i-1}\bigcap \lbIn{v_i}{i-1} = \emptyset)$ \textbf{then} \hfill /*$\phoplabel{i-1}$*/

18 \ \ \ \ \quad\quad\quad \textbf{if} $(\sidd{v}\not\in \ahash{})$ \textbf{then} 

19 \ \ \ \quad\quad\quad\quad $n_{D}\leftarrow n_{D}+1$

20 \ \ \quad\quad\quad\quad $\ahash{}[\sidd{v}]\leftarrow (n_{D}, v, 1)$

21 \ \ \ \quad\quad\quad \textbf{else}


22 \ \ \quad\quad\quad\quad $\ahash{}[\sidd{v}].s \leftarrow \ahash{}[\sidd{v}].s+1$

23 \ \ \quad\quad\quad $\aset{i} \leftarrow \aset{i} \cup \{v\}$

24 \ \ \quad\quad\quad $\sidd{v} \leftarrow n_{D}$

25 \ \ \quad \quad  \textbf{else} stop expansion from $v$

26 \ \ \quad $\lambda \leftarrow 0$

27 \ \ \quad \textbf{foreach} $((id, a, s_A)\in \ahash{}, (id, d, s_D)\in \ahash{})$ \textbf{do}

28 \ \ \quad\quad    \textbf{if} $(\lbOut{a}{i-1}\bigcap \lbIn{d}{i-1} \neq \emptyset)$ \textbf{then}   \hfill /*$\phoplabel{i-1}$*/

29 \ \ \quad\quad\quad $\lambda \leftarrow \lambda + |s_A|\times|s_D|$

30 \ \ \quad $\phopNk{i} \leftarrow |\aset{i}|\times |\dset{i}|-1-\lambda$

31 \ \ \quad $\phopNum{i} \leftarrow \phopNum{i-1} + \phopNk{i}$

32 \ \ \quad $\alpha \leftarrow \phopNum{i}/\tcv{G}$ \hfill /*\rr of $\Sk{i}$*/

33 \ \ \quad \textbf{foreach} $(a\in \aset{i})$ \textbf{do} \hfill /*compute $\phoplabel{i}$*/

34 \ \ \quad\quad\quad $\lbOut{a}{i} \leftarrow \lbOut{a}{i-1}\bigcup \{v_i\}$

35 \ \ \quad \textbf{foreach} $(d\in \dset{i})$ \textbf{do}\hfill /*compute $\phoplabel{i}$*/

36 \ \ \quad\quad\quad $\lbIn{d}{i} \leftarrow \lbIn{d}{i-1}\bigcup \{v_{i}\}$

37 \ \textbf{return} $\alpha$ as \rr of $\Sk{k}$

\end{algorithm}

\stitle{Analysis:} Same as Algorithm~\ref{algorithm:baseline} and Algorithm~\ref{algorithm:icrRR}, Algorithm~\ref{algorithm:icrRRplus} completes Step-1 by performing both forward and backward \bfs from each \hnode, during which it first computes the two sets $\aset{i}$ and $\dset{i}$ in lines 1-25, and at the same time computes $\ahash{}$ and $\dhash{}$. Then, it computes the new \phop $\phoplabel{i}$ in lines 33-36. The time cost is $O(k^2(|V|+|E|)$.
Different with Algorithm~\ref{algorithm:baseline} and Algorithm~\ref{algorithm:icrRR}, the benefits of Algorithm~\ref{algorithm:icrRRplus} lies in Step-2. The cost of Step-2 for each \hnode is $O(i|\parta{i}||\partd{i}|)$. For $k$ \hnode, the cost is $O(\sum_{i\in [1,k]}i|\parta{i}||\partd{i}|)$.
Therefore, the time complexity of Algorithm~\ref{algorithm:icrRRplus} is $O(k^2(|V|+|E|)+\sum_{i\in [1,k]}i|\parta{i}||\partd{i}|)$.

Similar to Algorithm~\ref{algorithm:baseline}, we need to maintain the 2-hop labels w.r.t. at most $k$ \hnodes during the processing. As
$\aset{i}$,
$\dset{i}$, $\ahash{}$ and $\dhash{}$ are bounded by $V$, and $\aset{i}(\dset{i})$ can be used to store nodes of $\aset{i+1}(\dset{i+1})$,
the space complexity of Algorithm~\ref{algorithm:icrRR} is $O(k|V|)$.

By comparing Algorithm~\ref{algorithm:baseline}, Algorithm~\ref{algorithm:icrRR} and Algorithm~\ref{algorithm:icrRRplus}, we know that the difference of the three algorithms lies in how to compute the \rr, i.e., Step-2. In Table~\ref{table:comparison}, we show the comparison of their time and space complexities. For time complexity, we do not show the cost of Step-1, due to that for Step-1, the cost is same for all three algorithms. We will show in Experiment that the incRR$^+$ algorithm works much more efficiently than the blRR and incRR algorithms.

\begin{table}[t]
\caption{Comparison of time and space complexities, where $A = \cup_{i\in [1,k]} \aset{i}$, $D = \cup_{i\in [1,k]} \dset{i}$, and $\parta{i}(\partd{i})$ is the partition of $\aset{i}(\dset{i})$.}
\label{table:comparison}

\centering


\begin{tabular} {|c|l|l|} \hline

       Algorithm & Time Complexity of Step-2 & Space Complexity \\

\hline
\hline	blRR	          	& $O(k|A||D|)$	& $O(k|V|)$\\
\hline	incRR	          	&	$O(\sum_{i\in [1,k]}i|\aset{i}||\dset{i}|)$ & $O(k|V|)$ \\
\hline	incRR$^+$	        & $O(\sum_{i\in [1,k]}i|\parta{i}||\partd{i}|)$	~~~~~~~~~~~~~	& $O(k|V|)$ \\

\hline
\end{tabular}
\end{table}

\vspace{3mm}

\begin{example} \label{example:incremental3} Consider $G$ in Figure~\ref{graph:graph}. Assume that we want to compute the \rr of $\Sk{3}=\{v_1,v_2,v_3\}$.

For $v_1$, as it is the first processed node, there is no covered reachability relationship, thus we do not need to test any reachability relationship in lines 28. As $|\aset{1}|=4$, $|\dset{1}|=7$, thus $\phopNk{1}=|\aset{1}|\times |\dset{1}|-1=27$ in line 30 of Algorithm~\ref{algorithm:icrRRplus}.

For $v_2$, as both $|\parta{2}|=|\partd{2}|=1$, we only need to test one reachable query, i.e., whether $\query{v_3}{v_{10}}$ can be answered by $\phoplabel{1}$. As $\lbOut{v_3}{1}\cap \lbIn{v_{10}}{1}=\emptyset$, we know that $\phopNk{2}=|\aset{2}|\times |\dset{2}|-1=15$ in line 30 of Algorithm~\ref{algorithm:icrRRplus}.

For $v_3$, as shown by Figure~\ref{graph:hash}(c), we know that $\parta{3}=\{\{v_3, v_5\},\{v_4, v_6, v_{11}\}\}$ and $\partd{3}=\{\{v_3,v_8, v_{14}\},\{v_7, v_{9}\}\}$. In line 28, we only need to test $|\parta{3}|\times |\partd{3}|=2\times 2=4$ reachable queries.
As $v_4$ can reach $v_7$ can be answered by $\phoplabel{2}$, we know that all nodes in $\{v_4, v_6, v_{11}\}$ can reach every node in $\{v_7, v_{9}\}$ can be answered by $\phoplabel{2}$, thus $\lambda=6$ for $v_3$. In line 30, we know that $\phopNk{3}=|\aset{3}|\times |\dset{3}|-1-\lambda=5\times 5-1-6=18$.

Then, we know that $\phopNum{3}=\phopNk{1}+\phopNk{2}+\phopNk{3}=27+15+18=60$, and the \rr is $\alpha = \phopNum{3}/\tcv{G}=60/70=85.7$\%. And during the processing, the total number of tested reachability queries by Algorithm~\ref{algorithm:icrRRplus} is 5.

As a comparison, Algorithm~\ref{algorithm:icrRR} needs to test 16 reachability queries for $v_2$, due to that $|\aset{2}|=|\dset{2}|=4$. For $v_3$, Algorithm~\ref{algorithm:icrRR} needs to test 25 reachability queries, due to that $|\aset{3}|=|\dset{3}|=5$. The total number of tested reachability queries for Algorithm~\ref{algorithm:icrRR} is 41.

Consider Algorithm~\ref{algorithm:baseline}, $|\aset{1}\cup\aset{2}\cup\aset{3}|=8$, $|\dset{1}\cup\dset{2}\cup\dset{3}|=10$, thus the total number of tested reachability queries for Algorithm~\ref{algorithm:baseline} is 80 to get the \rr.
\hfill ~ $\Box$
\end{example}

\vspace{3mm}

\begin{table}[t]
\caption{Statistics of datasets, where $d=2|\E|/|\V|$ is the average degree of
  $\G$,  $\tcv{\cdot}$ is the average number of reachable nodes for nodes of $G$, and $n_t$ is the number of topological levels (the length of the longest path) of $G$.\label{table:statistics-realdatasets}}

\centering



\begin{tabular} {|@{}c@{}||r|r|r|r|r|} \hline
        Dataset   & $|\V|$   &$|\E|$  &$d$ & $\tcv{\cdot}$ & $n_t$\\
\hline\hline
\amaze	&	3,710 	&	3,600 	&	1.94 	&	639	&	16		\\ \hline
\human	&	38,811 	&	39,576 	&	2.04 	&	9	&	18		\\ \hline
\anthra	&	12,499 	&	13,104 	&	2.10 	&	12	&	16		\\ \hline
\agrocyc	&	12,684 	&	13,408 	&	2.11 	&	13	&	16		\\ \hline
\ecoo	&	12,620 	&	13,350 	&	2.12 	&	14	&	22		\\ \hline
\vchocyc	&	9,491 	&	10,143 	&	2.14 	&	14	&	21		\\ \hline
\kegg	&	3,617 	&	3,908 	&	2.16 	&	729	&	26		\\ \hline
\arxiv	&	6,000 	&	66,707 	&	22.24 	&	928	&	167		\\ \hline\hline
\mail	&	231,000 	&	223,004 	&	1.93 	&	11,698	&	7		\\ \hline
\lj	&	971,232 	&	1,024,140 	&	2.11 	&	206,907	&	24		\\ \hline
\web	&	371,764 	&	517,805 	&	2.79 	&	55,055	&	34		\\ \hline
\patten	&	1,097,775 	&	1,651,894 	&	3.01 	&	3	&	7		\\ \hline
\citeseerten	&	770,539 	&	1,501,126 	&	3.90 	&	70	&	36		\\ \hline
~\patfive~	&	1,671,488 	&	3,303,789 	&	3.95 	&	8	&	12		\\ \hline
\citeseerfive	&	1,457,057 	&	3,002,252 	&	4.12 	&	116	&	36		\\ \hline
\citeseerxx	&	6,540,401 	&	15,011,260 	&	4.59 	&	15,510	&	59		\\ \hline
\dbp	&	3,365,623 	&	7,989,191 	&	4.75 	&	83,659	&	146		\\ \hline
\pat	&	3,774,768 	&	16,518,947 	&	8.75 	&	1,544	&	32		\\ \hline
\twitter	&	18,121,168 	&	18,359,487 	&	2.03 	&	1,346,820	&	22		\\ \hline
\webuk	&	22,753,644 	&	38,184,039 	&	3.36 	&	3,417,930	&	2793		\\ \hline

\end{tabular}
\end{table}

\section{Experiment}
\label{section:experiment}

In this section, we show experimental results on \rr computation. The compared algorithms include \blrr, \incrr, and \incrrplus. Moreover, we show the impacts of \phop on reachability queries processing based on the state-of-the-art algorithm \fl~\cite{DBLP:conf/edbt/VelosoCJZ14} in terms of index size, index construction time, and query time.
We implemented all algorithms using C++ and compiled by G++ 6.2.0.
All experiments
were run on a PC with Intel(R) Core(TM) i5-3230M CPU @ 3.0 GHz CPU, 16 GB
memory, and Ubuntu 18.04.1 Linux OS.
For algorithms that run $\geq
24$ hours or exceed the memory limit (16GB), we will show their
results as ``--'' in the tables.

\stitle{Datasets:}
Table~\ref{table:statistics-realdatasets} shows the statistics of 20
real datasets, where the first eight are small datasets ($|\V|\leq 100,000$) downloaded from the same web page$\footnote{\label{dataset:google-p}https://code.google.com/archive/p/grail/downloads}$.
The following 12 datasets are large ones ($|\V|> 100,000$).
These datasets are usually used in
the recent works w.r.t. reachability queries processing~\cite{DBLP:journals/vldb/YildirimCZ12,DBLP:conf/cikm/YanoAIY13, DBLP:conf/sigmod/ChengHWF13,DBLP:conf/sigmod/ZhuLWX14,DBLP:journals/tkde/SuZWY17,DBLP:journals/pvldb/WeiYLJ14,DBLP:conf/edbt/VelosoCJZ14,DBLP:conf/icde/SeufertABW13,
DBLP:journals/pvldb/JinW13,DBLP:conf/sigmod/ZhouZYWCT17}.
Among these datasets, \amaze and \kegg are metabolic networks, \human, \anthra, \agrocyc, \ecoo, \vchocyc are graphs describing the genome and biochemical machinery of E. coli K-12 MG1655.
\mail{}\footnote{\label{dataset:stanford}http://snap.stanford.edu/data/index.html} is an email network.
\lj is an online social network soc-LiveJournal1$^{\text{\ref{dataset:stanford}}}$.
\web is a web graph web-Google\footnote{\label{dataset:ferrari}https://code.google.com/p/ferrari-index/downloads/list}.
\arxiv, \patten{}\footnote{\label{dataset:baidu}http://pan.baidu.com/s/1bpHkFJx} , \citeseerten{}\text{$^{\ref{dataset:baidu}}$}, \patfive{}\text{$^{\ref{dataset:baidu}}$},
\citeseerfive{}\text{$^{\ref{dataset:baidu}}$},
\citeseerxx{}$^{\text{\ref{dataset:google-p}}}$ and \pat{}$^{\text{\ref{dataset:google-p}}}$ (cit-Patents) are all citation networks.
\dbp\footnote{http://pan.baidu.com/s/1c00Jq5E} is a knowledge graph Dbpedia.
\twitter{}$^{\text{\ref{dataset:ferrari}}}$ is a \DAG
transformed from a large-scale social network obtained from a crawl
of twitter.com~\cite{DBLP:conf/icwsm/ChaHBG10}.
\webuk{}$^{\text{\ref{dataset:ferrari}}}$ is a \DAG of a web graph
dataset.
For these datasets, \mail, \lj, \web, and the first seven small graphs are directed graphs initially. We transform each of them into a \DAG by coalescing each strongly connected component into a node. Note that this can be done in linear
time~\cite{DBLP:journals/siamcomp/Tarjan72}. All other datasets are \DAGs initially.
The statistics in Table~\ref{table:statistics-realdatasets} are that of \DAGs.

\subsection{Reachability Ratio Computation}
\stitle{Reachability Ratio and Index Size:} We show the \rr (RR) and the index size ratio (ISR) of the 20 real datasets in Figure~\ref{graph:reach-ratio}, where ISR denotes the ratio of the size of \phop w.r.t. $k$ \hnodes over the size of the 2-hop labels w.r.t. all nodes. From Figure~\ref{graph:reach-ratio} we have the following observation.

First, we can classify all datasets into three categories according to the value of their \rr. The first kind of datasets (\textbf{D1}) include \amaze, \kegg, \mail, \lj, \web, \citeseerxx, \dbp, \twitter, and \webuk, for which the RR is more than 99\% even when $k=1$, and both the RR and ISR almost do not change with the increase of $k$.
The second kind of datasets (\textbf{D2}) includes \human, \anthra, \agrocyc, \ecoo, \vchocyc, and \arxiv, for which both RR and ISR will become larger with the increase of $k$.
The third kind of datasets (\textbf{D3}) includes \patten, \citeseerten, \patfive, \citeseerfive, and \pat, for which both RR and ISR are very small or even approach zero; and with the increase of $k$, both RR and ISR almost do not change.
The value of $k$, therefore, only affects the second kind of datasets, and the \rr is more than 80\% when $k\geq 16$ for all datasets of the second kind, which indicates that we may benefit from using \phop on datasets of both the first and second kinds.

Second, the storage space used to maintain \phop is small compared with the \rr value. For example, for the first kind of datasets, we can use about $1/4$ storage space (ISR~$\approx25$\%) to maintain more than 99\% (RR $>99$\%) reachability information.

\begin{figure*}[htbp]
\centering
\subfigure[\amaze]{
\begin{minipage}{4.5cm}
\includegraphics{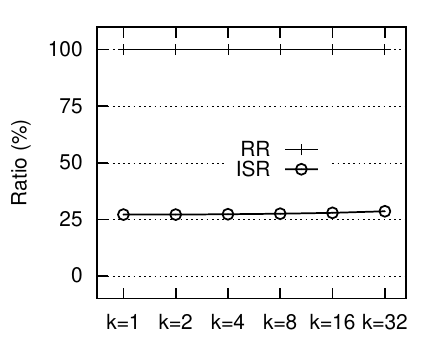}
\end{minipage}%
}%
\subfigure[\human]{
\begin{minipage}{4.5cm}
\includegraphics{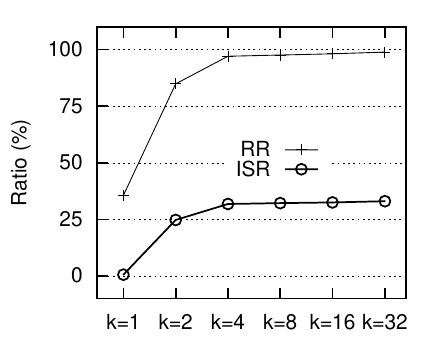}
\end{minipage}%
}%
\subfigure[\anthra]{
\begin{minipage}{4.5cm}
\includegraphics{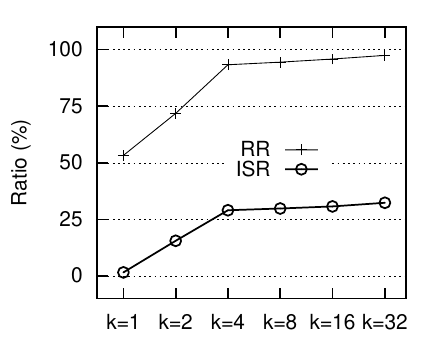}
\end{minipage}%
}%
\subfigure[\agrocyc]{
\begin{minipage}{4.5cm}
\includegraphics{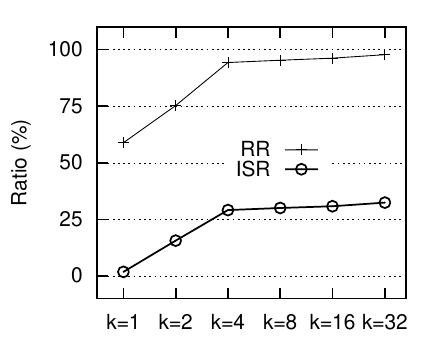}
\end{minipage}%
}%

\subfigure[\ecoo]{
\begin{minipage}{4.5cm}
\includegraphics{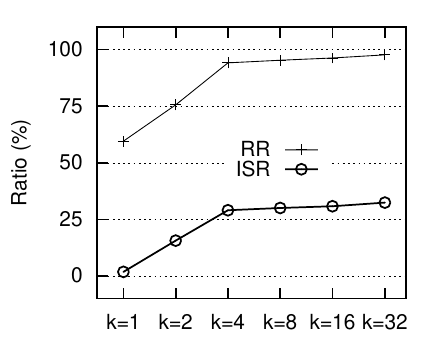}
\end{minipage}%
}%
\subfigure[\vchocyc]{
\begin{minipage}{4.5cm}
\includegraphics{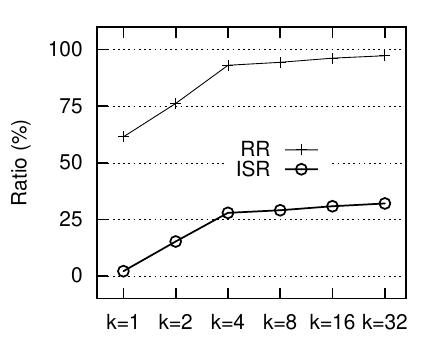}
\end{minipage}%
}%
\subfigure[\kegg]{
\begin{minipage}{4.5cm}
\includegraphics{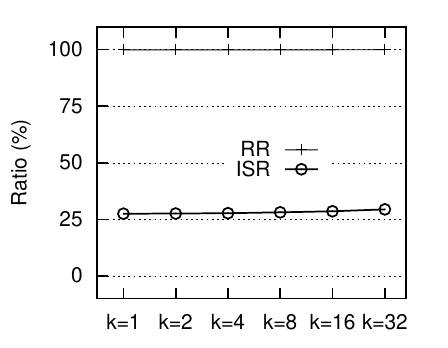}
\end{minipage}%
}%
\subfigure[\arxiv]{
\begin{minipage}{4.5cm}
\includegraphics{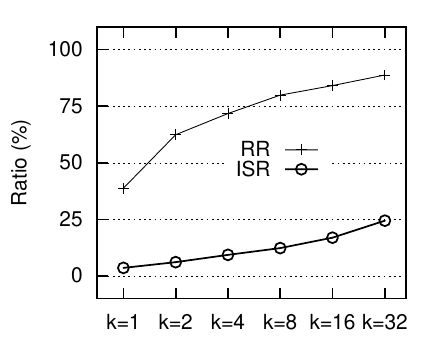}
\end{minipage}%
}%

\subfigure[\mail]{
\begin{minipage}{4.5cm}
\includegraphics{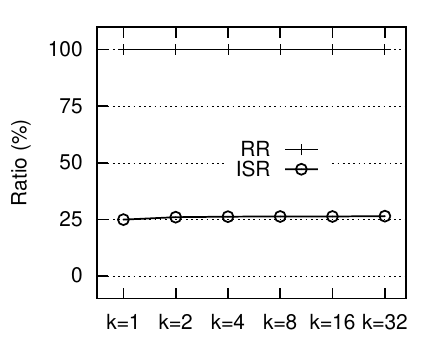}
\end{minipage}%
}%
\subfigure[\lj]{
\begin{minipage}{4.5cm}
\includegraphics{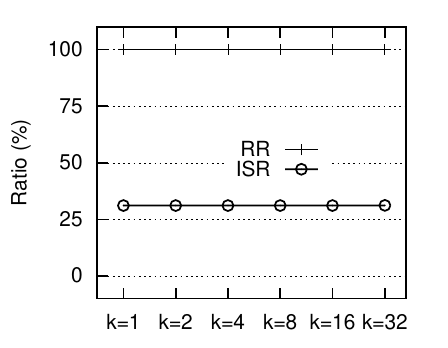}
\end{minipage}%
}%
\subfigure[\web]{
\begin{minipage}{4.5cm}
\includegraphics{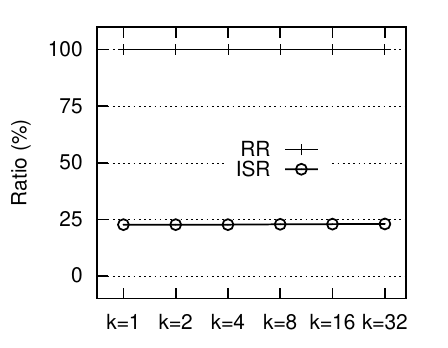}
\end{minipage}%
}%
\subfigure[\patten]{
\begin{minipage}{4.5cm}
\includegraphics{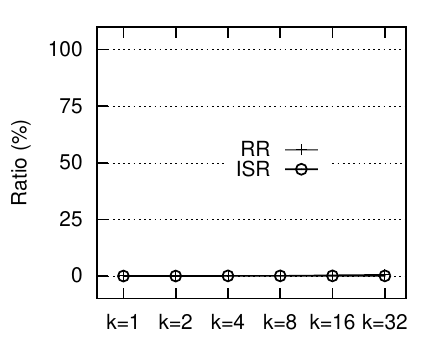}
\end{minipage}%
}%

\subfigure[\citeseerten]{
\begin{minipage}{4.5cm}
\includegraphics{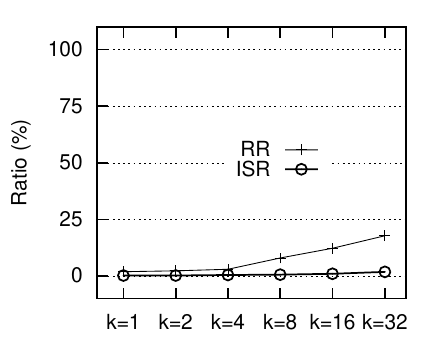}
\end{minipage}%
}%
\subfigure[\patfive]{
\begin{minipage}{4.5cm}
\includegraphics{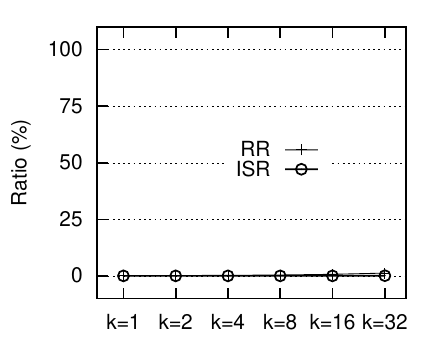}
\end{minipage}%
}%
\subfigure[\citeseerfive]{
\begin{minipage}{4.5cm}
\includegraphics{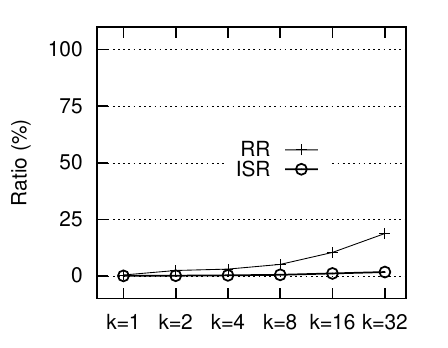}
\end{minipage}%
}%
\subfigure[\citeseerxx]{
\begin{minipage}{4.5cm}
\includegraphics{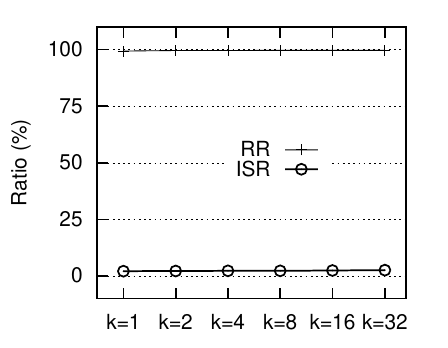}
\end{minipage}%
}%

\subfigure[\dbp]{
\begin{minipage}{4.5cm}
\includegraphics{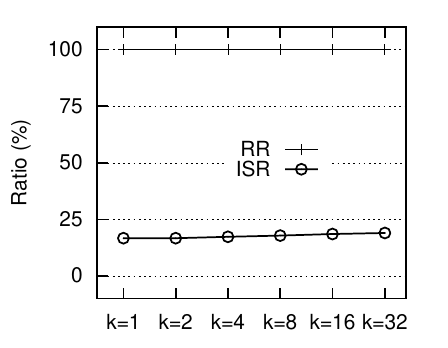}
\end{minipage}%
}%
\subfigure[\pat]{
\begin{minipage}{4.5cm}
\includegraphics{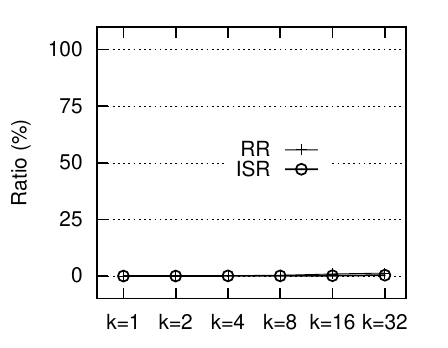}
\end{minipage}%
}%
\subfigure[\twitter]{
\begin{minipage}{4.5cm}
\includegraphics{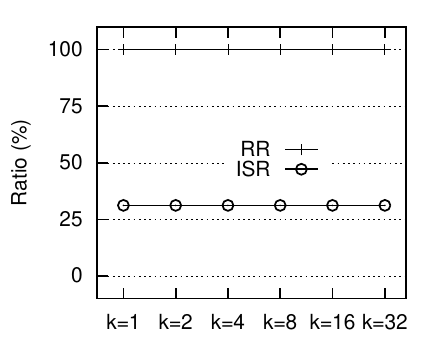}
\end{minipage}%
}%
\subfigure[\webuk]{
\begin{minipage}{4.5cm}
\includegraphics{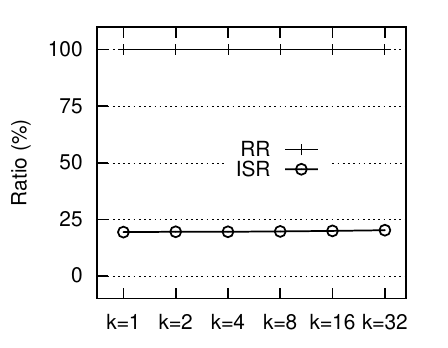}
\end{minipage}%
}%
\centering
\caption{Comparison of Reachability Ratio (RR) and Index Size Ratio (ISR), where ISR is ratio of the index size of \phop w.r.t. $k$ hop nodes over that of the total 2-hop label size w.r.t. all nodes.}\label{graph:reach-ratio}
\end{figure*}

\stitle{Running Time:} We show in Figure~\ref{graph:runningtime-reach-ratio} the comparison of running time for \rr computation, from which we have the following observations.

First, \incrrplus is much faster than both \blrr and \incrr on all datasets, and \incrr works faster than \blrr on most datasets. For instance, \incrrplus is faster than \blrr by more than two or three orders of magnitude on most datasets, and \incrr is ten times faster than \blrr on \amaze, \mail, \lj, \web, \citeseerxx, and \dbp. The reason can be explained as follows. On one hand, Figure~\ref{graph:reach-ratio} shows the \rr of different graphs w.r.t. different $k$. From Figure~\ref{graph:reach-ratio} we know that for \amaze, \human, \anthra, \agrocyc, \ecoo, \vchocyc, \kegg, \arxiv, \mail, \lj, \web, \citeseerxx, \dbp, \twitter and \webuk, the \rr is more than 80\% when $k=32$ for all datasets. On the other hand, according to the last to the second column of Table~\ref{table:statistics-realdatasets}, we know that the average number of reachable nodes for nodes of each graph is usually big. Therefore, the number of tested reachability queries by \blrr is significantly large. Even though \incrr can reduce the number of tested reachability queries, it still needs to test much more reachability queries than \incrrplus for some datasets. For example, consider the number of tested reachability queries on \kegg dataset when $k=32$. The tested number of reachability queries of \blrr (\incrr) is 100 (10) times more than that of \incrrplus. Moreover,
\blrr and \incrr cannot get the value of \rr on both \twitter and \webuk for $k\geq 2$ in limited time (24 hours), due to testing too many reachability queries.

\begin{figure*}[htbp]
\centering
\subfigure[\amaze]{
\begin{minipage}{4.5cm}
\includegraphics{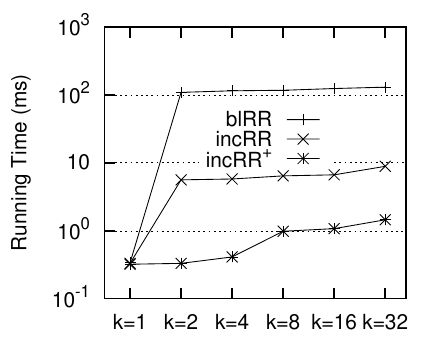}
\end{minipage}%
}%
\subfigure[\human]{
\begin{minipage}{4.5cm}
\includegraphics{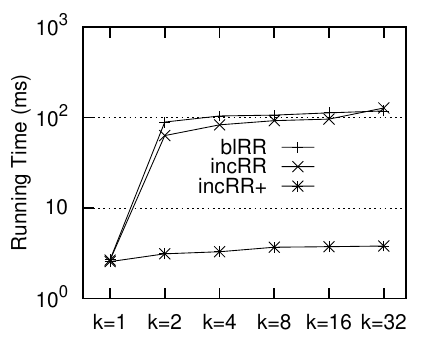}
\end{minipage}%
}%
\subfigure[\anthra]{
\begin{minipage}{4.5cm}
\includegraphics{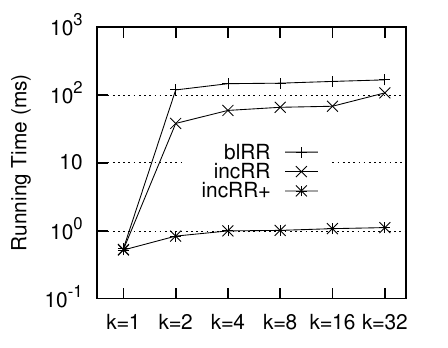}
\end{minipage}%
}%
\subfigure[\agrocyc]{
\begin{minipage}{4.5cm}
\includegraphics{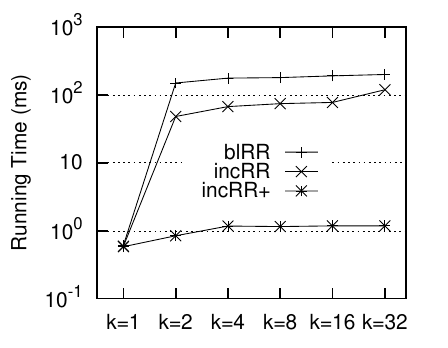}
\end{minipage}%
}%

\subfigure[\ecoo]{
\begin{minipage}{4.5cm}
\includegraphics{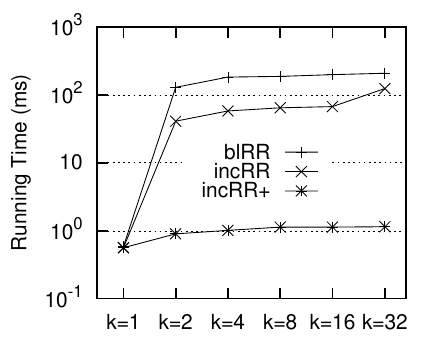}
\end{minipage}%
}%
\subfigure[\vchocyc]{
\begin{minipage}{4.5cm}
\includegraphics{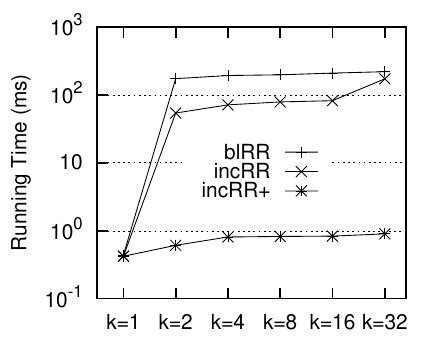}
\end{minipage}%
}%
\subfigure[\kegg]{
\begin{minipage}{4.5cm}
\includegraphics{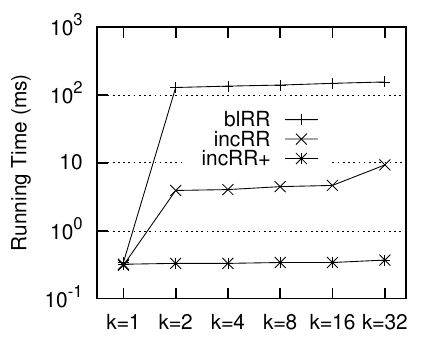}
\end{minipage}%
}%
\subfigure[\arxiv]{
\begin{minipage}{4.5cm}
\includegraphics{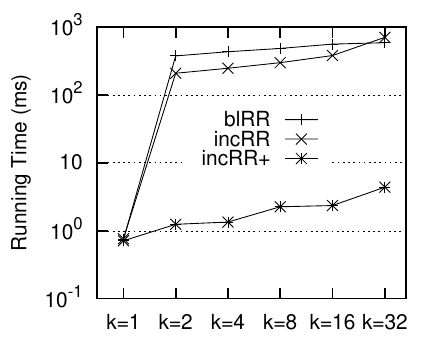}
\end{minipage}%
}%

\subfigure[\mail]{
\begin{minipage}{4.5cm}
\includegraphics{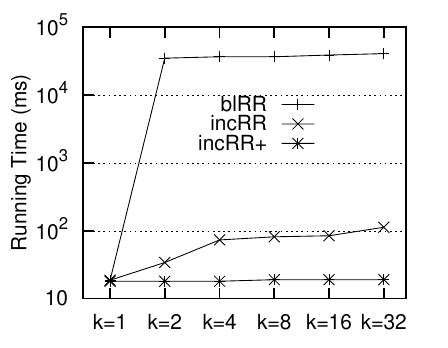}
\end{minipage}%
}%
\subfigure[\lj]{
\begin{minipage}{4.5cm}
\includegraphics{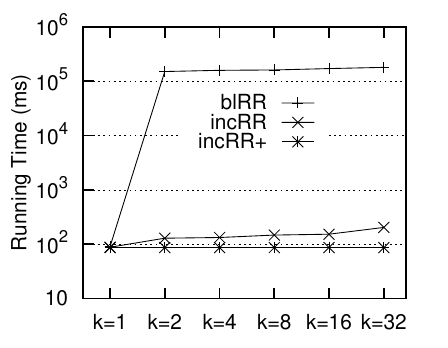}
\end{minipage}%
}%
\subfigure[\web]{
\begin{minipage}{4.5cm}
\includegraphics{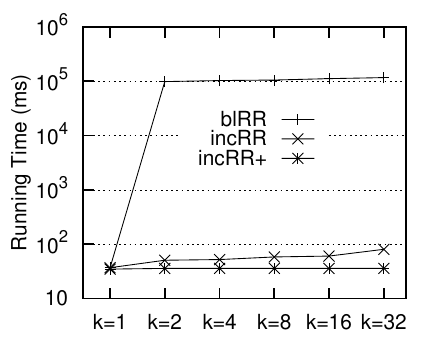}
\end{minipage}%
}%
\subfigure[\patten]{
\begin{minipage}{4.5cm}
\includegraphics{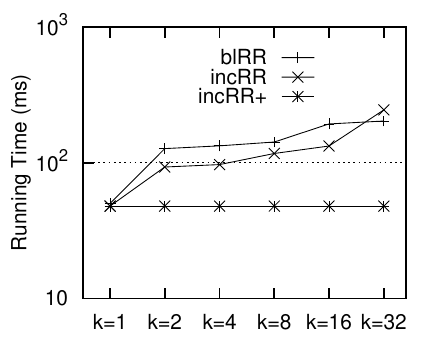}
\end{minipage}%
}%

\subfigure[\citeseerten]{
\begin{minipage}{4.5cm}
\includegraphics{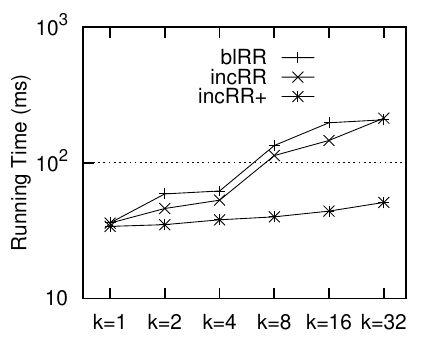}
\end{minipage}%
}%
\subfigure[\patfive]{
\begin{minipage}{4.5cm}
\includegraphics{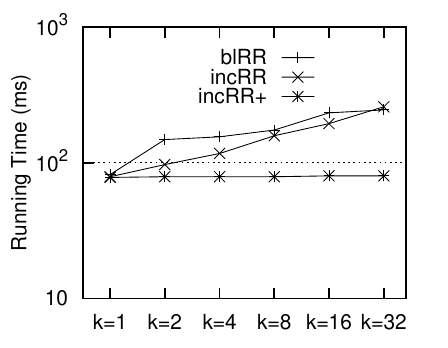}
\end{minipage}%
}%
\subfigure[\citeseerfive]{
\begin{minipage}{4.5cm}
\includegraphics{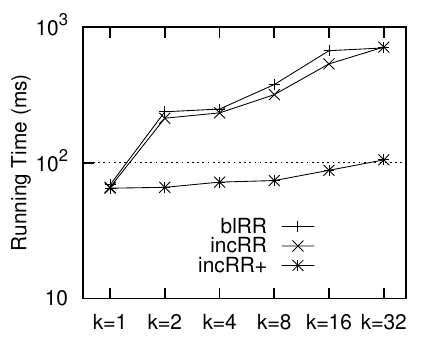}
\end{minipage}%
}%
\subfigure[\citeseerxx]{
\begin{minipage}{4.5cm}
\includegraphics{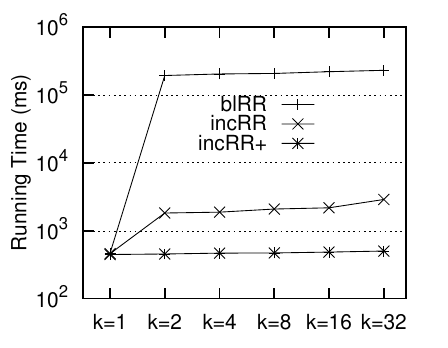}
\end{minipage}%
}%

\subfigure[\dbp]{
\begin{minipage}{4.5cm}
\includegraphics{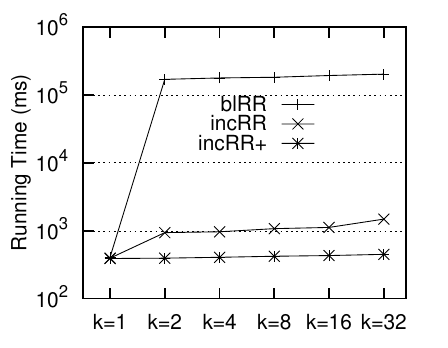}
\end{minipage}%
}%
\subfigure[\pat]{
\begin{minipage}{4.5cm}
\includegraphics{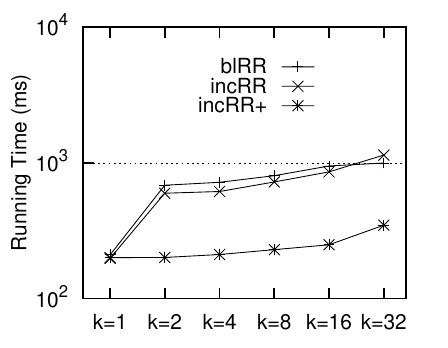}
\end{minipage}%
}%
\subfigure[\twitter]{
\begin{minipage}{4.5cm}
\includegraphics{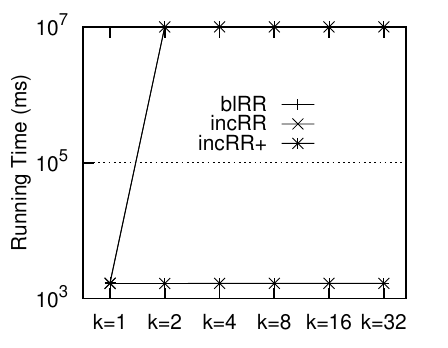}
\end{minipage}%
}%
\subfigure[\webuk]{
\begin{minipage}{4.5cm}
\includegraphics{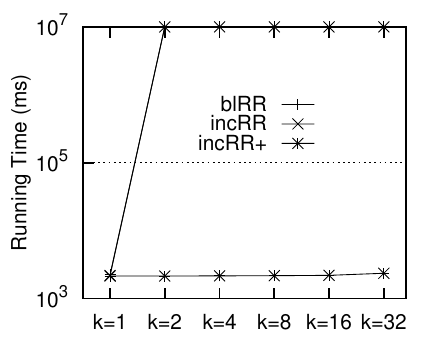}
\end{minipage}%
}%
\centering
\caption{Running Time of different algorithms on \rr computation (ms).}
\label{graph:runningtime-reach-ratio}
\end{figure*}

Second, both \blrr and \incrr work efficiently on datasets where the \rr is small. For instance, from Figure~\ref{graph:runningtime-reach-ratio} we know that \incrrplus is faster than \blrr and \incrr by less than ten times on \patten, \citeseerten, \patfive, \citeseerfive, and \pat. The reason lies in that for these datasets, the \rr is very small according to Figure~\ref{graph:reach-ratio}, which means that for all algorithms, the number of tested reachability queries is much less than other datasets, therefore does not need to consume more time.

It is worth noting that when $k=1$, the three algorithms consume similar time. The reason is that when $k=1$, for the first \hnode $v_1$, after we get $\aset{v_1}$ and $\dset{v_1}$, we immediately know the number of reachability queries covered by $\phoplabel{1}$ is $|\aset{v_1}|\times |\dset{v_1}|-1$, and therefore do not need to actually test any reachability queries.

By the above experimental results, we know that our \incrrplus algorithm can be used to efficiently compute the \rr for a given dataset, which brings us a chance to determine whether we should use \phop to facilitate reachability queries processing.

\subsection{Reachability Queries Processing}

In this section, we combine \phop with the state-of-the-art algorithm, namely \feline~\cite{DBLP:conf/edbt/VelosoCJZ14} (abbreviated as \fl), to show the impact of \phop on reachability queries processing, in terms of index size, index construction time and query time. The experimental results are shown, respectively, in Table~\ref{table:reach-indexsize}, Table~\ref{table:reach-indexingtime} and Table~\ref{table:reach-querytime}, where \fl-$k$ denotes the \fl algorithm combined with \phop that are generated based on $k$ \hnodes. Hence, \fl-0 is the \fl algorithm without \phop. Note that for reachability queries processing, we do not set $k=1,2,4,8$, due to that when $k=16$, we only need to use one integer as a bit-vector for each node $v$ to represent both $\lbOut{v}{16}$ and $\lbIn{v}{16}$.

\stitle{Index Size:} Table~\ref{table:reach-indexsize} shows the impacts of $k$ on index size, from which we know that with the increase of $k$, the index size will increase accordingly. For example, for \webuk dataset, the index size of \fl-128 is more than two times bigger than that of \fl-0 on all datasets. The reason is obvious. The larger the value of $k$, the more the space we need to maintain the \phop.

\stitle{Index Construction Time:} Figure~\ref{table:reach-indexingtime} shows the impacts of $k$ on index construction time, from which we know that with the increase of $k$, we need more time for index construction. Note that \phop can be constructed efficiently, and the increased time for index construction could be omitted, due to that index construction is a one-time activity performed off-line for reachability queries processing.

\stitle{Query Time:} We report the query time about equal workload, which contains 1,000,000 reachability queries for each dataset. The equal workload consists of 50\% reachable queries and 50\% unreachable queries. The reason that we use equal workload is:
using completely random queries is heavily skewed towards unreachable queries~\cite{DBLP:journals/pvldb/YildirimCZ10, DBLP:journals/vldb/YildirimCZ12}, which is highly unlikely for the real workload as the node pair in a query tends
to have a certain connection~\cite{DBLP:conf/sigmod/JinRDY12}. Here, unreachable queries are generated by sampling node pairs with the same probability until we reach the required number of unreachable queries by testing each query using the \fl algorithm. For reachable queries, we cannot choose them randomly by sampling the \tc, because \tc~ computation suffers from high time and space complexity, we cannot get it within limited time and memory size for large graphs.
To this problem, we randomly
pick a node $u$ in each iteration, then randomly select an out-neighbor $v$
recursively until $v$ has no out-neighbor. Then, we have a path $p$ from $u$.
At last, we randomly select a node $v\neq u$ from $p$ to get a reachable query $\querytrue{u}{v}$.
This operation will be continued until we reach the required number
of reachable queries.

We show the comparison of query time for \fl-0 to \fl-128 in Table~\ref{table:reach-querytime}, from which we have the following observations.

First, \fl-16 and \fl-32 usually need the least time on the first kind of datasets \textbf{D1}, including
\amaze, \kegg, \mail, \lj, \web, \citeseerxx, \dbp, \twitter and \webuk,
where the \rr is more than 99\% even when $k=1$. For these datasets, although the index size becomes larger and the index construction time becomes longer than that of \fl-0, we use the least cost to achieve significant improvements. For example, compared with \fl-0, \fl-16 and \fl-32 use about 1.5 times index size and 1.2 times index construction time to achieve more than 1,000 times improvements on query time.

Second, \fl-128 suffers from the largest index size (about 3 times bigger than \fl-0) and longest index construction time (about 1.3 times longer than \fl-0), but achieves the best query performance on the second kind of datasets \textbf{D2}, due to that on these datasets, the \rr will become larger with the increase of $k$. These datasets include \human, \anthra, \agrocyc, \ecoo, \vchocyc and \arxiv.

Third, for the third kind of datasets \textbf{D3}, including \patten, \citeseerten, \patfive, \citeseerfive, and \pat, the \rr is very small or even approach zero, and almost does not change with the increase of $k$. For these datasets, \fl-0 works best and the use of \phop cannot bring us any positive results. For example, compared with \fl-0 on \patfive, the index size of \fl-128 is 2.6 times bigger than \fl-0, and the index construction time and query time of \fl-128 are 1.04 and 1.5 times longer than that of \fl-0.

At last, we choose one dataset from each kind and show the trend of its query time w.r.t. $k$ in Figure~\ref{graph:rchquerytime}, from which we can give out the suggestions on how to use \phop: (1) For the first kind of datasets \textbf{D1}, we highly recommend using \phop with $k=16$ to process reachability queries, due to that we can speed up reachability queries answering significantly by affording only a little more index size and index construction time. (2) For the second kind of datasets \textbf{D2}, we also recommend using \phop, due to that we can speed up reachability queries answering by \phop. But for the value of $k$, it depends on your concerns on how much you could and would like to afford for the increased index size and index construction time. In general, the larger the value of $k$, the less the query time, but the more the index construction time and the bigger the index size. (3) For the third kind of datasets \textbf{D3}, we do not recommend using \phop to process reachability queries.

\begin{table}[t]
\caption{Comparison of the index size (MB).\label{table:reach-indexsize}}

\centering



\begin{tabular} {|c||r|r|r|r|r|} \hline
        Dataset   & \fl-0   & \fl-16  & \fl-32 & \quad \fl-64 & \quad \fl-128  \\
\hline\hline
\amaze	&	0.07 	&	0.08 	&	0.10 	&	0.13 	&	0.18 		\\ \hline
\human	&	0.74 	&	0.89 	&	1.04 	&	1.33 	&	1.92 		\\ \hline
\anthra	&	0.24 	&	0.29 	&	0.33 	&	0.43 	&	0.62 		\\ \hline
\agrocyc	&	0.24 	&	0.29 	&	0.34 	&	0.43 	&	0.63 		\\ \hline
\ecoo	&	0.24 	&	0.29 	&	0.34 	&	0.43 	&	0.62 		\\ \hline
\vchocyc	&	0.18 	&	0.22 	&	0.25 	&	0.32 	&	0.47 		\\ \hline
\kegg	&	0.07 	&	0.08 	&	0.10 	&	0.12 	&	0.18 		\\ \hline
\arxiv	&	0.11 	&	0.14 	&	0.16 	&	0.20 	&	0.30 		\\ \hline\hline
\mail	&	4.4 	&	5.3 	&	6.2 	&	7.9 	&	11.5 		\\ \hline
\lj	&	18.5 	&	22.2 	&	25.9 	&	33.3 	&	48.2 		\\ \hline
\web	&	7.1 	&	8.5 	&	9.9 	&	12.8 	&	18.4 		\\ \hline
\patten	&	20.9 	&	25.1 	&	29.3 	&	37.7 	&	54.4 		\\ \hline
\citeseerten	&	14.7 	&	17.6 	&	20.6 	&	26.5 	&	38.2 		\\ \hline
\patfive	&	31.9 	&	38.3 	&	44.6 	&	57.4 	&	82.9 		\\ \hline
\citeseerfive	&	27.8 	&	33.3 	&	38.9 	&	50.0 	&	72.3 		\\ \hline
\citeseerxx	&	124.7 	&	149.7 	&	174.6 	&	224.5 	&	324.3 		\\ \hline
\dbp	&	64.2 	&	77.0 	&	89.9 	&	115.5 	&	166.9 		\\ \hline
\pat	&	72.0 	&	86.4 	&	100.8 	&	129.6 	&	187.2 		\\ \hline
\twitter	&	345.6 	&	414.8 	&	483.9 	&	622.1 	&	898.6 		\\ \hline
\webuk	&	434.0 	&	520.8 	&	607.6 	&	781.2 	&	1,128.4 		\\ \hline

\end{tabular}
\end{table}

\begin{table}[t]
\caption{Comparison of the index construction time (ms).\label{table:reach-indexingtime}}

\centering



\begin{tabular} {|c||r|r|r|@{~~}r|@{~~}r|} \hline
        Dataset   & \fl-0   & \fl-16  & \fl-32 & \fl-64 & \fl-128  \\
\hline\hline
\amaze	&	1.03 	&	1.10 	&	1.35 	&	1.41 	&	1.49 		\\ \hline
\human	&	9.01 	&	10.01 	&	11.50 	&	11.65 	&	11.67 		\\ \hline
\anthra	&	2.89 	&	2.99 	&	3.67 	&	3.72 	&	3.85 		\\ \hline
\agrocyc	&	2.96 	&	3.27 	&	3.63 	&	4.19 	&	3.96 		\\ \hline
\ecoo	&	3.08 	&	3.27 	&	3.91 	&	6.05 	&	4.00 		\\ \hline
\vchocyc	&	2.22 	&	2.37 	&	3.16 	&	3.47 	&	3.65 		\\ \hline
\kegg	&	1.11 	&	1.20 	&	1.62 	&	1.65 	&	1.39 		\\ \hline
\arxiv	&	4.71 	&	4.61 	&	6.34 	&	7.25 	&	8.44 		\\ \hline\hline
\mail	&	81.3 	&	77.2 	&	91.0 	&	86.5 	&	87.1 		\\ \hline
\lj	&	325.1 	&	327.3 	&	383.9 	&	376.3 	&	387.8 		\\ \hline
\web	&	178.5 	&	177.9 	&	203.8 	&	202.7 	&	207.6 		\\ \hline
\patten	&	801.3 	&	803.3 	&	862.0 	&	832.0 	&	862.1 		\\ \hline
\citeseerten	&	376.3 	&	385.3 	&	415.5 	&	419.1 	&	437.6 		\\ \hline
\patfive	&	1,517.8 	&	1,495.2 	&	1,518.4 	&	1,566.9 	&	1,577.2 		\\ \hline
\citeseerfive	&	775.6 	&	784.8 	&	840.0 	&	828.8 	&	899.1 		\\ \hline
\citeseerxx	&	4,063.9 	&	4,053.7 	&	4,500.5 	&	4,441.1 	&	4,562.0 		\\ \hline
\dbp	&	2,264.3 	&	2,371.0 	&	2,588.6 	&	2,607.1 	&	2,598.4 		\\ \hline
\pat	&	5,022.8 	&	5,152.1 	&	5,328.3 	&	5,400.1 	&	5,372.3 		\\ \hline
\twitter	&	6,287.6 	&	6,446.3 	&	7,236.2 	&	7,233.9 	&	7,719.2 		\\ \hline
\webuk	&	8,689.7 	&	8,774.7 	&	9,945.0 	&	9,991.6 	&	10,366.0 		\\ \hline

\end{tabular}
\end{table}

\begin{table}[t]
\caption{Comparison of the query time (ms).\label{table:reach-querytime}}

\centering



\begin{tabular} {|c||r|r|r|@{~~}r|@{~~}r|} \hline
        Dataset   & \fl-0   & \fl-16  & \fl-32 & \fl-64 & \fl-128  \\
\hline\hline
\amaze	&	592 	&	30 	&	28 	&	30 	&	30 		\\ \hline
\human	&	190 	&	37 	&	36 	&	35 	&	33 		\\ \hline
\anthra	&	133 	&	32 	&	32 	&	31 	&	32 		\\ \hline
\agrocyc	&	137 	&	33 	&	31 	&	31 	&	31 		\\ \hline
\ecoo	&	143 	&	35 	&	36 	&	36 	&	30 		\\ \hline
\vchocyc	&	126 	&	33 	&	31 	&	32 	&	28 		\\ \hline
\kegg	&	533 	&	45 	&	39 	&	40 	&	36 		\\ \hline
\arxiv	&	1,105 	&	594 	&	566 	&	554 	&	511 		\\ \hline\hline
\mail	&	8,091 	&	21 	&	24 	&	26 	&	40 		\\ \hline
\lj	&	50,811 	&	37 	&	42 	&	55 	&	83 		\\ \hline
\web	&	39,902 	&	58 	&	59 	&	71 	&	96 		\\ \hline
\patten	&	248 	&	277 	&	285 	&	304 	&	412 		\\ \hline
\citeseerten	&	347 	&	377 	&	409 	&	427 	&	509 		\\ \hline
\patfive	&	405 	&	475 	&	482 	&	528 	&	605 		\\ \hline
\citeseerfive	&	452 	&	485 	&	489 	&	533 	&	633 		\\ \hline
\citeseerxx	&	162,696 	&	575 	&	544 	&	554 	&	587 		\\ \hline
\dbp	&	25,344 	&	99 	&	104 	&	128 	&	176 		\\ \hline
\pat	&	13,180 	&	13,493 	&	13,566 	&	13,752 	&	13,898 		\\ \hline
\twitter	&	\----	&	88 	&	91 	&	117 	&	176 		\\ \hline
\webuk	&	\----	&	3,197 	&	3,230 	&	3,478 	&	3,701 		\\ \hline

\end{tabular}
\end{table}

\begin{figure}
  \centering
\includegraphics{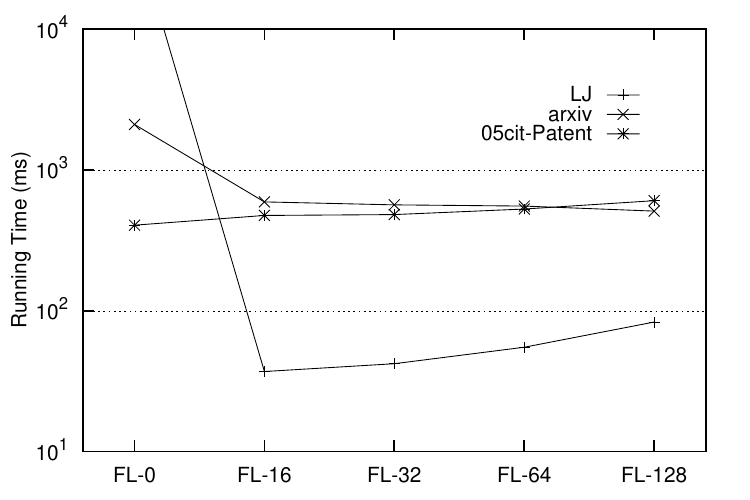}

\setlength{\abovecaptionskip}{-0.5pt}
\caption{Impacts of $k$ on query time (ms) over different datasets.}
\label{graph:rchquerytime}
\end{figure}

\section{Conclusion}
\label{section:conclusion}

Partial 2-hop label is a useful pruning technique for reachability queries processing. In practice, it may be powerful to answer most queries by a larger \rr for some graphs, but for other graphs, its pruning ability may not be as powerful as expected, or even makes query performance degenerated on some graphs, due to small \rr. In this paper, we aim at figuring out an important problem: whether we should use \phop for reachability queries processing for a given graph. To solve this problem, we formally defined the \rr problem and proposed a set of algorithms for efficient \rr computation. Our first experimental results show that compared with the baseline algorithm, our optimized algorithm can efficiently compute the \rr for a given graph.
Our second experimental results show that by combining \phop with an existing reachability algorithm, the query performance has different trends with the increase of the number of \hnodes $k$. And based on the second experimental results, we finally give out our findings on whether we should use \phop for reachability queries processing. Specifically, (1) for datasets with large \rr, \phop should be used with $k=16$; (2) for datasets with small \rr, we do not recommend using \phop; and (3) for the remaining datasets, \phop can be used, and users can determine $k$'s value themselves according to their requirements on index size, index construction time and query time.

\section{Acknowledgments}

This work was partly supported by the grants from the Natural Science Foundation of Shanghai (No. 20ZR1402700), and from the Natural Science Foundation of China (No.: 61472339, 61572421, 61272124). The authors would like to thank the anonymous referees for their insightful and valuable comments. 

\bibliographystyle{ieeetr}
\bibliography{sigproc}

\end{document}